\documentclass{lmcs}
\pdfoutput=1

\usepackage{lastpage}
\lmcsdoi{15}{2}{12}
\lmcsheading{}{\pageref{LastPage}}{}{}%
{May~18,~2018}{May~08,~2019}{}


\usepackage[T1]{fontenc}
\usepackage[utf8]{inputenc}
\usepackage{amssymb,amsmath}
\usepackage{graphicx, color}
\usepackage{soul}

\usepackage{mathtools}

\usepackage[pagewise]{lineno}
\usepackage{xspace}
\usepackage{ifmtarg}
\usepackage{stmaryrd}
\usepackage{tikz}
\usetikzlibrary{arrows,patterns, decorations.pathmorphing}%
\usepackage{hyperref}

\newenvironment{theorem}{\begin{thm}}{\end{thm}}
\newenvironment{lemma}{\begin{lem}}{\end{lem}}
\newenvironment{corollary}{\begin{cor}}{\end{cor}}
\newenvironment{proposition}{\begin{prop}}{\end{prop}}

\newenvironment{definition}{\begin{defi}}{\end{defi}}

\newcommand  {\myclass} [1]  {\ensuremath{\textsc{#1}}}

\newcommand{\StaClass}[1]{\myclass{#1}\xspace}

\newcommand{\DynClass}[1]{\myclass{Dyn#1}\xspace}

\newcommand  {\myproblem} [1] {\textsc{#1}}

\newcommand  {\problem}[1] {\myproblem{#1}}

\newcommand     {\LOGSPACE}     {\myclass{LOGSPACE}}

\newcommand     {\NL}   {\myclass{NL}}

\newcommand     {\AC}   {\myclass{AC}}

\newcommand{\FO}{\StaClass{FO}}
\newcommand{\FOar}{\StaClass{FO$(+,\!\times\!)$}}
\newcommand{\MSO}{\StaClass{MSO}}
\newcommand{\GSO}{\StaClass{GSO}}

\newcommand{\CQ}[1][]{\StaClass{CQ}}
\newcommand{\UCQ}[1][]{\StaClass{UCQ}}
\newcommand{\CQneg}[1][]{\StaClass{CQ\ensuremath{^{\mneg}}}}
\newcommand{\UCQneg}[1][]{\StaClass{UCQ\ensuremath{^{\mneg}}}}

\newcommand{\Ind}[1]{\ensuremath{\text{IND}[#1]}}

\newcommand{\mneg}{\neg} %

\newcommand{\DynFO}{\DynClass{FO}}

\newcommand{\DynFOar}{\DynClass{FO$(+,\times)$}}

\newcommand{\mtext}[1]{\textsc{#1}}

\providecommand {\calA}      {{\mathcal A}\xspace}
\providecommand {\calB}      {{\mathcal B}\xspace}
\providecommand {\calC}      {{\mathcal C}\xspace}
\providecommand {\calD}      {{\mathcal D}\xspace}

\providecommand {\calI}      {{\mathcal I}\xspace}

\providecommand {\calO}      {{\mathcal O}\xspace}
\providecommand {\calP}      {{\mathcal P}\xspace}

\providecommand {\calS}      {{\mathcal S}\xspace}

\providecommand {\calY}      {{\mathcal Y}\xspace}

\newcommand{\N}{\ensuremath{\mathbb{N}}}

\newcommand{\R}{\ensuremath{\mathbb{R}}}

\newcommand{\bigO}{\ensuremath{\mathcal{O}}}
\newcommand{\tpl}{\bar}

\newcommand{\restrict}[2]{#1\mspace{-3mu}\upharpoonright \mspace{-3mu}#2}

\newcommand{\df}{\ensuremath{\mathrel{\smash{\stackrel{\scriptscriptstyle{
    \text{def}}}{=}}}} \;}

\makeatletter %
\newcommand{\auxramsey}[4]{
  \@ifmtarg{#1}{
    \@ifmtarg{#4}{
      \ensuremath{R(#2; #3)}
    }{
      \ensuremath{R^#4(#2; #3)}
    }
   }{
    \@ifmtarg{#4}{
      \ensuremath{R_{#1}(#2; #3)}
    }{
      \ensuremath{R^#4_{#1}(#2; #3)}
    }
  }
}

    \newenvironment{proofsketch}{\begin{proof}[Proof sketch.]}{\end{proof}}
    \newenvironment{proofof}[1]{\begin{proof}[Proof (of #1).]}{\end{proof}}
\newcommand{\arity}{\ensuremath{\text{Ar}}}

\newcommand{\struc}{\calS}

\newcommand{\inp}{\ensuremath{\calI}\xspace}
\newcommand{\aux}{\ensuremath{\calA}\xspace}

\newcommand{\query}{\ensuremath{q}}

\newcommand{\adom}{\ensuremath{\text{adom}}}

\newcommand{\BIT}{\ensuremath{\text{BIT}}}

\newcommand{\state}{\ensuremath{\struc}\xspace}

\newcommand{\prog}{\ensuremath{\calP}\xspace}

\newcommand{\uf}[4]{
  \@ifmtarg{#4}{
    \ensuremath{\phi^{#1}_{#2}(#3)}
   }{
    \ensuremath{\phi^{#1}_{#2}(#3; #4)}
  }
}
\newcommand{\huf}[4]{
  \@ifmtarg{#4}{
    \ensuremath{\widehat{\phi}^{#1}_{#2}(#3)}
   }{
    \ensuremath{\widehat{\phi}^{#1}_{#2}(#3; #4)}
  }
}

\newcommand{\ufb}[4]{
  \@ifmtarg{#4}{
    \ensuremath{\psi^{#1}_{#2}(#3)}
   }{
    \ensuremath{\psi^{#1}_{#2}(#3; #4)}
  }
}

  \makeatletter %
  \newcommand{\ufsubstitute}[5]{
    \@ifmtarg{#5}{
      \ensuremath{\phi^{#2}_{#3}[#1](#4)}
    }{
      \ensuremath{\phi^{#2}_{#3}[#1](#4; #5)}
    }
  }

  \makeatletter %
  
  \makeatletter %

\newcommand{\ut}[4]{
  \@ifmtarg{#4}{
    \ensuremath{t^{#1}_{#2}(#3)}
   }{
    \ensuremath{t^{#1}_{#2}(#3; #4)}
  }
}

\newcommand{\ite}[3]{
  \@ifmtarg{#1}{
    \ensuremath{\mtext{ITE}}
   }{
    \mtext{ITE}(#1,#2,#3)  
  }
}

\newcommand{\mf}[3]{
  \@ifmtarg{#3}{
    \ensuremath{\mu_{#1}(#2)}
   }{
    \ensuremath{\mu_{#1}(#2; #3)}
  }
}

\newcommand{\mfos}[4]{
  \@ifmtarg{#4}{
    \ensuremath{{#1}_{#2}(#3)}
   }{
    \ensuremath{{#1}_{#2}(#3; #4)}
  }
}

\newcommand{\Qreach}{\ensuremath{\query_{\text{Reach}}}}

 \renewcommand{\restrict}[2]{#1[#2]}

\usetikzlibrary{arrows, shapes, calc}%

\tikzstyle{mnode}=[
  circle,
  fill=\mnodefillcolor, 
  draw=\mnodedrawcolor,
  minimum size=6pt, 
  inner sep=0pt
]

\tikzstyle{mnodeinvisible}=[
  minimum size=6pt, 
  inner sep=0pt
]

\tikzstyle{invisible}=[
  minimum size=0pt, 
  inner sep=0pt
]

\tikzstyle{invisiblel}=[
  minimum size=10pt, 
  inner sep=0pt
]

\tikzstyle{invisibleEdge}=[
  transparent
]

\tikzstyle{nameNode}=[
  font=\scriptsize
]

\tikzstyle{namingNode}=[
  font=\normalsize
]

\tikzstyle{mEdge}=[
  -latex', %
  thick, 
  shorten >=3pt, 
  shorten <=3pt,
  draw=black!80,
]

\tikzstyle{dDashedEdge}=[
  -latex', %
  thick, 
  shorten >=3pt, 
  shorten <=3pt,
  draw=black!80,
  dashed
]

\tikzstyle{dEdge}=[
  -latex', %
  thick, 
  shorten >=1pt, 
  shorten <=1pt,
  draw=black!80,
]

\tikzstyle{dhEdge}=[
  -latex', %
  thick, 
  shorten >=3pt, 
  shorten <=3pt,
  draw=black!80,
]
\tikzstyle{uEdge}=[
  thick, 
  shorten >=3pt, 
  shorten <=3pt,
  draw=black!80,
]
\tikzstyle{uhEdge}=[
  thick, 
  shorten >=3pt, 
  shorten <=3pt,
  draw=black!80,
]

\tikzstyle{cEdge}=[
  ultra thick, 
  shorten >=3pt, 
  shorten <=3pt,
  draw=black!80,
]

\tikzstyle{dotsEdge}=[
  very thick, 
  loosely dotted, 
  shorten >=7pt, 
  shorten <=7pt
]

\tikzstyle{class rectangle}=[
  draw=black,
  inner sep=0.2cm,
  rounded corners=5pt,
  thick
]

\tikzstyle{mline}=[
  draw=black,
  inner sep=0.2cm,
  rounded corners=5pt,
  thick
]

\tikzstyle{mainclass rectangle}=[
  draw=blue,
  inner sep=0.2cm,
  rounded corners=5pt,
  very thick
]

\newcommand{\mnodedrawcolor}{black!80}
\newcommand{\mnodefillcolor}{black!40}

\pgfdeclarelayer{background}
\pgfdeclarelayer{substructure}
\pgfdeclarelayer{edges}
\pgfdeclarelayer{foreground}
\pgfsetlayers{background,substructure,edges,main,foreground}

\tikzstyle{background rectangle}=[
  fill=black!10,
  draw=black!20,
  line width = 5pt,
  inner sep=0.4cm,
  outer sep=0.4cm,
  rounded corners=5pt
] 
\begin{document}

\title[A Strategy for Dynamic Programs: Start over and Muddle through]{A Strategy for Dynamic Programs: Start over and Muddle through}
\titlecomment{This article is the full version of  \cite{DattaMSVZ17}. The authors acknowledge the financial support by the DAAD-DST grant ``Exploration of New Frontiers in Dynamic Complexity''. The first and the second authors were partially funded by a grant from Infosys foundation. The second author was partially supported by a TCS PhD fellowship. The last three authors acknowledge the financial support by DFG grant SCHW 678/6-2 on ``Dynamic Expressiveness of Logics''.}

\author[S.~Datta]{Samir Datta\rsuper{a}}	%
\address{\lsuper{a}Chennai Mathematical Institute \& UMI ReLaX, Chennai, India}	%
\email{sdatta@cmi.ac.in}  %

\author[A.~Mukherjee]{Anish Mukherjee\rsuper{b}}	%
\address{\lsuper{b}Chennai Mathematical Institute, Chennai, India}	%
\email{anish@cmi.ac.in}  %

\author[T.~Schwentick]{Thomas Schwentick\rsuper{c}}	%
\address{\lsuper{c}TU Dortmund University, Dortmund, Germany}	%
\email{\{thomas.schwentick,nils.vortmeier,thomas.zeume\}@tu-dortmund.de}  %

\author[N.~Vortmeier]{Nils Vortmeier\rsuper{c}}	%

\author[T.~Zeume]{Thomas Zeume\rsuper{c}}	%

  \begin{abstract}
 In the setting of \DynFO, dynamic programs update the stored result of a query whenever the underlying data changes.
 This update is expressed in terms of first-order logic.
 We introduce a strategy for constructing dynamic programs that utilises periodic computation of auxiliary data from scratch and the ability to maintain a query for a limited number of change steps.
  We show that if some program can maintain a query for $\log n$
change steps after an $\AC^1$-computable initialisation, it can be
maintained by a first-order dynamic program as well, i.e., in \DynFO. 
As an application, it is shown that decision and optimisation problems
defined by monadic second-order (\MSO) 
formulas are in \DynFO, if only change sequences that produce graphs of
bounded treewidth are allowed. To establish this result, 
a Feferman-Vaught-type composition theorem for \MSO is established that might be
useful in its own right.  
  \end{abstract}
  
  \maketitle
  
  \section{Introduction}\label{section:introduction}

Each time a database is changed, any previously computed and stored result of a fixed query might become outdated. However, when the change is small, it is plausible that the new query result is highly related to the old result. 
In that case it might be more efficient to use previously computed information for obtaining the new answer to the query instead of recomputing the query result from scratch.
A theoretical framework for studying when the result of a query over relational databases can be updated in a declarative fashion was formalised by Patnaik and Immerman \cite{PatnaikI97}, and Dong, Su, and Topor \cite{DongST95}. In their formalisation, a dynamic program consists of a set of logical formulas that update a query result after the insertion or deletion of a tuple. The formulas may use additional auxiliary relations that, of course, need to be updated as well. The queries maintainable in this way via first-order formulas constitute the dynamic complexity class $\DynFO$.

Recent work has confirmed that \DynFO is a quite powerful class, since
it captures, e.g., the reachability query for directed graphs
\cite{DattaKMSZ18}, and even allows for more complex change operations than single-tuple changes \cite{SchwentickVZ18, DattaMVZ18}. 

In this article we introduce a general strategy for dynamic programs
that further underscores the expressive power of \DynFO.
All prior results for \DynFO yield dynamic programs that are able to maintain a query for arbitrary long sequences of changes.
Even if this is not (known to be) possible for a given query, in a practical scenario it might still be favourable to maintain the query result dynamically for a bounded number of changes, then to apply a more complex algorithm that recomputes certain auxiliary information from scratch, such that afterwards the query can again be maintained for some time, and so on.

Here we formalise this approach.
Let $\calC$ be a complexity class and $f \colon \N \rightarrow \N$ a function. 
A query $\query$ is called $(\calC,f)$-maintainable, if there is a dynamic program (with first-order definable updates) that, starting from some input structure $\calA$ and auxiliary relations computed in $\calC$ from $\calA$, can answer $\query$ for $f(|\calA|)$ many steps, where $|\calA|$ denotes the size of the universe of $\calA$. 

We feel that this notion might be interesting in its own
right. However, in this article we concentrate on the case where $\calC$ is (uniform) $\AC^1$ and
$f(n)=\log n$. The class $\AC^1$ contains all queries that can be computed by a (uniform) circuit of depth $\calO(\log n)$ that uses polynomially many $\wedge$-, $\vee$-, and $\neg$-gates, where $\wedge$- and $\vee$-gates may have unbounded fan-in.
We show that $(\AC^1,\log n)$-maintainable queries are
actually in \DynFO, and thus can be maintained for arbitrary long change sequences. 

We apply this insight to show that all queries and optimisation problems definable in monadic second-order
logic (\MSO) are in \DynFO for (classes of) structures of bounded
treewidth, by proving that they are $(\AC^1,\log n)$-maintainable. The same can be said about guarded second-order logic (\GSO), simply because it is expressively equivalent to \MSO on such classes \cite{Courcelle94}.
This
implies that decision problems like \problem{3-Colourability} or
\problem{HamiltonCycle} as well as optimisation problems like
\problem{VertexCover} and \problem{DominatingSet} are in \DynFO, for
such classes of structures.
This result is therefore a dynamic version of Courcelle's Theorem which states that all problems definable in (certain extensions of) \MSO can statically be solved in linear time for graphs of bounded treewidth \cite{Courcelle90}.

The proof that \MSO-definable queries are $(\AC^1,\log n)$-maintainable
on structures of bounded treewidth
makes use of a Feferman–Vaught-type composition theorem for \MSO which might be
useful in its own right.  
 
The result that $(\AC^1,\log n)$-maintainable queries are
in \DynFO comes with a technical restriction: in a nutshell, it holds
for queries that are invariant under insertion of (many) isolated
elements. We call such queries \emph{almost domain independent} and
refer to Section~\ref{section:framework} for a precise definition.

We emphasise that the main technical challenge in maintaining
\MSO-queries on graphs of bounded treewidth is that tree
decompositions might change drastically after an edge
insertion, and can therefore not be maintained incrementally in any
obvious way. In particular, the result does not simply follow from the
\DynFO-maintainability of regular tree languages shown in \cite{GeladeMS12}.
We circumvent this problem by periodically recomputing a new
tree decomposition (this can be done in logarithmic space \cite{ElberfeldJT10} and thus in $\AC^1$) and by
showing that \MSO-queries can be maintained for $\bigO(\log n)$ many
change operations, even if they make the tree decomposition invalid.

\subsection*{Contributions}
We briefly summarise the contributions described above.
In this article, we introduce the notion of $(\calC, f)$-maintainability and show that, amongst others, (almost domain independent) $(\AC^1,\log n)$-maintainable queries are in \DynFO.
We show that \MSO-definable decision problems and optimisation problems are $(\AC^1,\log
  n)$-maintainable and therefore in \DynFO, for structures of bounded
treewidth.
These proofs make use of a Feferman–Vaught-type composition theorem for \MSO logic.

\subsection*{Related work}

The simulation-based technique for proving that $(\AC^1,\log n)$-maintainable queries are in \DynFO is inspired by proof techniques from \cite{DattaKMSZ15} and  \cite{SchwentickVZ18}. As mentioned above, in \cite{GeladeMS12} it has been shown that tree languages, i.e.~MSO on trees, can be maintained in \DynFO. Independently, the maintenance of \MSO definable queries on graphs of bounded treewidth is also studied in \cite{BouyerJM17}, though in the restricted setting where the tree decomposition stays the same for all changes.
A static but parallel version of Courcelle's Theorem is given in \cite{ElberfeldJT10}: every \MSO-definable problem for graphs of bounded treewidth can be solved with logarithmic space.

\subsection*{Organisation}
Basic terminology is recalled in Section \ref{section:preliminaries},
followed by a short introduction into dynamic complexity in
Section \ref{section:framework}. In Section \ref{section:technique} we
introduce the notion of  $(\calC, f)$-maintainability and show that
$(\AC^1,\log n)$-maintainable queries are in $\DynFO$. A glimpse on
the proof techniques for proving that \MSO  queries are in
\DynFO for graphs of bounded treewidth is given in Section
\ref{section:3col} via the example \problem{3-Colourability}. The proof of the
general results is presented in Section \ref{section:mso}. An
extension to optimisation problems can be found in Section
\ref{section:mso-opt}. 

This article is the full version of \cite{DattaMSVZ17}.

  \section{Preliminaries}\label{section:preliminaries}
We now introduce some notation and notions regarding logics, graph theory and complexity theory.
We assume familiarity with first-order logic \FO and other notions from finite model theory~\cite{ImmermanDC,Libkin04}. %

\subsection*{Relational structures}
In this article we consider finite relational structures over relational signatures $\Sigma = \{R_1, \ldots, R_\ell,c_1,\ldots,c_m\}$, where each $R_i$ is a relation symbol with a corresponding arity $\arity(R_i)$, and each $c_j$ is a constant symbol.
A $\Sigma$-structure $\calA$ consists of a finite \emph{domain} $A$, also called the \emph{universe} of $\calA$, a relation $R_i^\calA \subseteq A^{\arity(R_i)}$, and a constant $c_j^\calA \in A$, for each $i \in \{1, \ldots, \ell\}, j \in \{1, \ldots, m\}$. 
The \emph{active domain} $\adom(\calA)$  of a structure $\calA$ contains all elements used in some tuple or as some constant of $\calA$. 
For a set $B \subseteq A$ that contains all constants and a relation $R$, the restriction $\restrict{R}{B}$ of $R$ to $B$ is the relation $R \cap B^{\arity(R)}$. The structure $\restrict{\calA}{B}$ induced by $B$ is the structure obtained from $\calA$ by restricting the domain and all relations to $B$. 

Sometimes, especially in Section~\ref{section:framework}, we consider relational structures as relational databases. This terminology is common in the context of dynamic complexity due to its original motivation from relational databases.  Also for this reason, dynamic complexity classes  defined later will be defined as classes of queries of arbitrary arity, and not as a class of decision problems. However, we will mostly consider queries over structures with a single binary relation symbol $E$, that is, queries on graphs.

We will often use structures $\calA$ with a linear order $\le$ on the universe $A$, and compatible ternary relations encoding arithmetical operations  $+$ and $\times$ or, alternatively, a binary relation encoding the relation $\BIT = \{(i,j) \mid $ the $j$-th bit in the binary representation of $i$ is $1\}$. The linear order in particular allows us to identify $A$ with the first $|A|$ natural numbers. We write $\FO(+,\times)$ or $\FO(\BIT)$ to emphasise that we allow  first-order formulas to use such additional relations.\footnote{The question of $<$-invariance (c.f.~\cite{Libkin04}) will not be relevant in the context of this article as a specific relation $\leq$ will be available in the structure.} We also use that $\FO(+,\times)=\FO(\BIT)$ \cite{ImmermanDC}.

\subsection*{Tree decompositions and treewidth}
A \emph{tree decomposition} $(T,B)$ of $G$ consists of a (rooted, directed) tree $T = (I,F,r)$, with (tree) nodes $I$, (tree) edges $F$, a distinguished root node $r \in I$, and a function $B \colon I \rightarrow 2^V$ such that
\begin{enumerate}
 \item[(1)] the set $\{i \in I \mid v \in B(i)\}$ is non-empty for each node $v \in V$,
 \item[(2)] there is an $i \in I$ with $\{u,v\} \subseteq B(i)$ for each edge $(u, v) \in E$, and
 \item[(3)] the subgraph $T[\{i \in I \mid v \in B(i)\}]$ is connected for each node $v \in V$.
\end{enumerate}

We refer to the number of children of a node $i$ of $T$ as its \emph{degree}, and to the set $B(i)$ as its \emph{bag}. We denote the parent node of $i$ by $p(i)$.
The \emph{width} of a tree decomposition is defined as the maximal size of a bag minus $1$. The \emph{treewidth} of a graph $G$ is the minimal width among all tree decompositions of $G$.
A tree decomposition is \emph{$d$-nice}, for some $d \in \N$, if 
\begin{enumerate}
 \item[(1)] $T$ has depth at most $d \log n$,
 \item[(2)] the degree of the nodes is at most $2$, and
 \item[(3)] all bags are distinct. 
\end{enumerate} 
Often we do not make the constant $d$ explicit and just speak of \emph{nice} tree decompositions.

Later we will use that tree decompositions can be transformed into nice tree decompositions with slightly increased width. This is formalized in the following lemma, whose proof is an adaption of \cite[Lemma 3.1]{ElberfeldJT10}.
\begin{lemma} \label{prelims:decompositions:nice}
For every $k \in \N$ there is a constant $d \in \N$ such that for every graph of treewidth $k$, a $d$-nice tree decomposition of width $4k+5$ can be computed in logarithmic space.
\end{lemma}
\begin{proof}
Let $G$ be a graph of treewidth $k$. By \cite[Lemma 3.1]{ElberfeldJT10} a tree decomposition $(T,B)$ of width $4k+3$ can be computed in logarithmic space, such that each non-leaf node has degree 2 and the depth is at most $d \log n$, for a constant $d$ that only depends on $k$. 
To obtain a tree decomposition with distinct bags, we compose this algorithm with three further algorithms, each reading a tree decomposition $(T,B)$ and transforming it into a tree decomposition $(T',B')$ with a particular property. Since each of the four algorithms requires only logarithmic space, the same holds for their composition. 

The first transformation algorithm produces a tree decomposition, in which for each leaf bag $i$ it holds $B(i)\not\subseteq B(p(i))$. In particular, after this transformation, each bag of a leaf node $i$ contains some graph node $u(i)$ that does not appear in any other bag. This transformation inspects each node $i$ separately in a bottom-up fashion, and removes it if (1) $B(i)\subseteq B(p(i))$ and (2) every bag below $i$ is a subset of $B(i)$. Clearly, logarithmic space suffices for this.

The second transformation (inductively) removes an inner node $i$ of degree 1 with a child~$i'$ and inserts an edge between $B(p(i))$ and $B(i')$ whenever $B(i)\subseteq B(p(i))$ or $B(i)\subseteq B(i')$ holds. For this transformation, only one linear chain of nodes in $T$ has to be considered at any time and therefore logarithmic space suffices again. Clearly, the connectivity property is not affected by these deletions.

The third transformation adds to every bag of an inner node $i$ the nodes $u(i_1)$ and $u(i_2)$, guaranteed to exist by the first transformation, of the  leftmost and rightmost leaf nodes $i_1$ and $i_2$ of the subtree rooted at $i$, respectively. Here, we assume the children of every node to be ordered by the representation of $T$ as input to the algorithm. 
After this transformation, each node of degree 2 has a different bag than its two children thanks to the addition of  $u(i_1)$ and $u(i_2)$. Each node of degree 1 has a different bag than its child, since this was already the case before (and to both of them the same two nodes might have been added). Altogether, all bags are pairwise distinct and the bag sizes have increased by at most 2.

We emphasise that, whenever a leftmost graph node  $u(i_1)$ is added to $B(i)$, it is also added to all bags of nodes on the path from $i$ to $i_1$ and therefore the connectivity property is not corrupted. It is easy to see that the third transformation can also be carried out in logarithmic space. 

The three presented algorithms never increase the depth of a tree decomposition, so the final result is a $d$-nice tree decomposition for $G$ of width $4k+5$.
\end{proof}
In this paper we only consider nice tree decompositions, and due to property (3) of these decompositions we can identify bags with nodes from $I$.

For two nodes $i,i'$ of $I$, we write $i\preceq i'$ if $i'$ is in the subtree of $T$ rooted at $i$ and  $i \prec i'$ if, in addition, $i'\not=i$. 
 A \emph{triangle} $\delta$ of $T$ is a triple $(i_0,i_1,i_2)$ of nodes from $I$ such that $i_0\preceq i_1$,  $i_0\preceq i_2$, and (1) $i_1=i_2$ or (2) neither $i_1\preceq i_2$ nor $i_2\preceq i_1$.
In case of (2) we call the triangle \emph{proper}, in case of (1) \emph{unary}, unless $i_0=i_1=i_2$ in which we call it \emph{open}  (see Figure \ref{figure:triangle} for an illustration).
  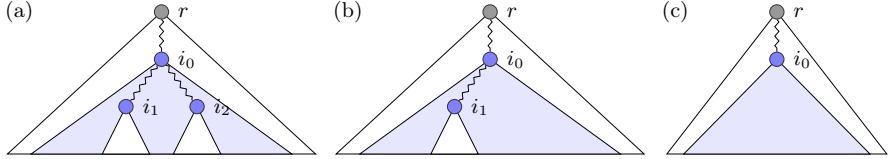
\begin{figure}[t] 
    \begin{center}
    \scalebox{0.9}{
          \begin{tikzpicture}[
            xscale=0.7,
            yscale=0.7,
            font=\footnotesize,
          ]
                \node (tmp) at (-3, 3) {(a)};
                \node (r) at (0,3)[mnode, label=right:$r$] {};
                \node[mnode, fill=blue!50] (i0) at (0,2)[label=right:$i_0$] {};
                \node[mnode, fill=blue!50] (i1) at (-0.75,1)[label=right:$i_1$] {};
                \node[mnode, fill=blue!50] (i2) at (0.75,1)[label=right:$i_2$] {};

                \node (a) at (-3.25,0)[invisible] {};
                \node (b) at (-2.75,0)[invisible] {};
                \node (c) at (-1.25,0)[invisible] {};
                \node (d) at (-0.25,0)[invisible] {};
                \node (e) at (0.25,0)[invisible] {};
                \node (f) at (1.25,0)[invisible] {};
                \node (g) at (2.75,0)[invisible] {};
                \node (h) at (3.25,0)[invisible] {};
      
                \draw [decoration={zigzag,segment length=4,amplitude=.9},line join=round] (r) edge[decorate]  (i0);
                \draw [decoration={zigzag,segment length=4,amplitude=.9},line join=round] (i0) edge[decorate]  (i1);
                \draw [decoration={zigzag,segment length=4,amplitude=.9},line join=round] (i0) edge[decorate]  (i2);

                \begin{pgfonlayer}{background}
                \draw[] (a.center) -- (h.center) -- (r.center) -- cycle;
                \draw[fill=blue!10] (b.center) -- (g.center) -- (i0.center) -- cycle;
                \draw[fill=white] (c.center) -- (d.center) -- (i1.center) -- cycle;
                \draw[fill=white] (e.center) -- (f.center) -- (i2.center) -- cycle;
               \end{pgfonlayer}
      \end{tikzpicture}
          \begin{tikzpicture}[
            xscale=0.7,
            yscale=0.7,
            font=\footnotesize,
          ]
                \node (tmp) at (-3, 3) {(b)};
                \node (r) at (0,3)[mnode, label=right:$r$] {};
                \node[mnode, fill=blue!50] (i0) at (0,2)[label=right:$i_0$] {};
                \node[mnode, fill=blue!50] (i1) at (-0.75,1)[label=right:$i_1$] {};

                \node (a) at (-3.25,0)[invisible] {};
                \node (b) at (-2.75,0)[invisible] {};
                \node (c) at (-1.25,0)[invisible] {};
                \node (d) at (-0.25,0)[invisible] {};
                \node (e) at (0.25,0)[invisible] {};
                \node (f) at (1.25,0)[invisible] {};
                \node (g) at (2.75,0)[invisible] {};
                \node (h) at (3.25,0)[invisible] {};
      
                \draw [decoration={zigzag,segment length=4,amplitude=.9},line join=round] (r) edge[decorate]  (i0);
                \draw [decoration={zigzag,segment length=4,amplitude=.9},line join=round] (i0) edge[decorate]  (i1);

                \begin{pgfonlayer}{background}
                \draw[] (a.center) -- (h.center) -- (r.center) -- cycle;
                \draw[fill=blue!10] (b.center) -- (g.center) -- (i0.center) -- cycle;
                \draw[fill=white] (c.center) -- (d.center) -- (i1.center) -- cycle;
               \end{pgfonlayer}
      \end{tikzpicture}
          \begin{tikzpicture}[
            xscale=0.5,
            yscale=0.7,
            font=\footnotesize,
          ]
                \node (tmp) at (-3, 3) {(c)};
                \node (r) at (0,3)[mnode, label=right:$r$] {};
                \node[mnode, fill=blue!50] (i0) at (0,2)[label=right:$i_0$] {};

                \node (a) at (-3.25,0)[invisible] {};
                \node (b) at (-2.75,0)[invisible] {};
                \node (c) at (-1.25,0)[invisible] {};
                \node (d) at (-0.25,0)[invisible] {};
                \node (e) at (0.25,0)[invisible] {};
                \node (f) at (1.25,0)[invisible] {};
                \node (g) at (2.75,0)[invisible] {};
                \node (h) at (3.25,0)[invisible] {};
      
                \draw [decoration={zigzag,segment length=4,amplitude=.9},line join=round] (r) edge[decorate]  (i0);

                \begin{pgfonlayer}{background}
                \draw[] (a.center) -- (h.center) -- (r.center) -- cycle;
                \draw[fill=blue!10] (b.center) -- (g.center) -- (i0.center) -- cycle;
               \end{pgfonlayer}
      \end{tikzpicture}

    }
    \caption{Illustration of (a) a proper triangle $(i_0, i_1, i_2)$, (b) a unary triangle $(i_0, i_1, i_1)$, and (c) an open triangle $(i_0, i_0, i_0)$. The blue 
shaded area is the part of the tree contained in the triangle.}\label{figure:triangle}
    \end{center}%
  \end{figure}
 The subtree $T(\delta)$ \emph{induced} by a triangle consists of all nodes $j$  of $T$ for which the following holds: 
 \begin{enumerate}
\item[(i)]$i_0 \preceq j$, 
\item[(ii)] if $i_0\prec i_1$ then  $i_1\not\prec j$, and 
\item[(iii)] if $i_0\prec i_2$ then  $i_2\not\prec j$.
\end{enumerate}
That is, for a proper or unary triangle, $T(\delta)$ contains all nodes of the subtree rooted at $i_0$ which are not below $i_1$ or $i_2$. For an open triangle $\delta=(i_0,i_0,i_0)$,  $T(\delta)$ is just the subtree rooted at~$i_0$. 

Each triangle $\delta$ induces a subgraph $G(\delta)$ of $G$  as follows: $V(\delta)$ is the union of all  bags of~$T(\delta)$. By $B(\delta)$ we denote the set $B(i_0)\cup B(i_1)\cup B(i_2)$ of \emph{interface nodes} of $V(\delta)$. All other nodes in $V(\delta)$ are called \emph{inner nodes}.
The edge set of $G(\delta)$ consists of all edges of $G$ that involve at least one inner node of $V(\delta)$.

\subsection*{MSO-logic and MSO-types}

\MSO is the extension of first-order logic that allows existential and universal quantification over set variables $X,X_1,\ldots$. The \emph{(quantifier) depth} of an \MSO formula is the maximum nesting depth of (second-order and first-order) quantifiers in the syntax tree of the formula.  

For a signature $\Sigma$ and a natural number $d\ge 0$, the \emph{depth-$d$ \MSO type} of a $\Sigma$-structure~$\calA$ is defined as the set of all \MSO sentences $\varphi$ over $\Sigma$ of quantifier depth at most $d$, for which $\calA\models\varphi$ holds. 
We also define the notion of types for structures with additional constants and \MSO formulas with free variables. Let  $\calA$ be a $\Sigma$-structure and $\tpl v=(v_1,\ldots,v_m)$ a tuple of elements from $\calA$. We write $(\calA,\tpl v)$ for the structure over $\Sigma\cup\{c_1,\ldots,c_m\}$ which interprets $c_i$ as $v_i$, for every $i\in\{1,\ldots,m\}$. For a  set $\calY$ of first-order and second-order variables and an assignment $\alpha$ for the variables of $\calY$, the depth-$d$ \MSO type of  $(\calA,\tpl v, \alpha)$ is the set of \MSO formulas with free variables from $\calY$ of depth $d$ that hold in  $(\calA,\tpl v, \alpha)$. 
For every depth-$d$ \MSO type $\tau$, there is a depth-$d$ \MSO formula $\alpha_\tau$ that is true in exactly the structures and for those assignments of type $\tau$.

The logic \emph{guarded second-order logic} (\GSO) extends \MSO by \emph{guarded} second-order quantification.
Thus, it syntactically allows to quantify over non-unary relation variables. However, this quantification is semantically restricted: a tuple $\tpl t = (a_1, \ldots, a_m)$ can only occur in a quantified relation, if all elements from $\{a_1, \ldots, a_m \}$ occur together in some tuple of the structure, in which the formula is evaluated.

For more background on \MSO logic and types, readers might consult, e.g., \cite{Libkin04}.

\subsection*{Complexity classes and descriptive characterisations}
Our main result refers to the complexity class (uniform) $\AC^1$. It contains all queries that can be computed by (families of uniform) circuits of depth $\bigO(\log n)$, consisting of polynomially many ``and'', ``or'' and ``not'' gates, where ``and'' and ``or'' gates may have unbounded fan-in.
It contains the classes  $\LOGSPACE$ and $\NL$, and it can be characterised as the class $\Ind{\log n}$  of problems that can be expressed by applying a first-order formula $\bigO(\log n)$ times \cite[Theorem 5.22]{ImmermanDC}. Here, $n$ denotes the size of the universe and the formulas can use built-in relations $+$ and~$\times$. 
More generally, this characterisation is also valid for the analogously defined classes $\AC[f(n)]$ and $\Ind{f(n)}$, where the depth of the circuits and the number of applications of the first-order formula is $f(n)$, respectively, for some function $f \colon \N \to \N$. Technically, the function $f$ needs to be \emph{first-order constructible}, that is, there has to be a $\FOar$ formula $\psi_f(\tpl x)$ such that $\calA \models \psi_f(\tpl a)$ if and only if $\tpl a$ is a base-$n$ representation of $f(n)$, for any ordered structure $\calA$ with domain $\{0, \ldots, n-1\}$. 

Our proofs often assume that $\log n$ is a natural number, but they can be easily adapted to the general case.

  \section{Dynamic Complexity} \label{section:framework}
\renewcommand{\aux}{\ensuremath{Aux}}

We briefly repeat the essentials of dynamic complexity, closely following \cite{SchwentickZ16, DattaKMSZ18}. %

The goal of a dynamic program is to answer a given query on an \emph{input database} subjected to changes that insert or delete single tuples. The program may use an auxiliary data structure represented by an \emph{auxiliary database} over the same domain. Initially, both input and auxiliary database are empty; and the domain is fixed during each run of the program. 

A dynamic program has a set of update rules that specify how auxiliary relations are updated after a change of the input database. An \emph{update rule} for updating an auxiliary relation $T$ is
basically a formula $\varphi$. As an example, if $\varphi(\tpl x, \tpl y)$ is the update rule for auxiliary relation $T$ under insertions into input relation $R$, then the new version of $T$ after insertion of a tuple $\tpl a$ to $R$ is $T \df \{ \tpl b \mid (\inp, \aux) \models \varphi(\tpl a, \tpl b)\}$ where $\inp$ and $\aux$ are the current input and auxiliary databases.
For a state $\state = (\inp, \aux)$ of the dynamic program $\prog$ with input database $\inp$ and auxiliary database $\aux$ we denote the state of the program after applying a sequence $\alpha$ of changes by $\prog_\alpha(\state)$. 
The dynamic program $\prog$ \emph{maintains} a $k$-ary query $\query$ if, for each non-empty sequence $\alpha$ of changes and each empty input structure $\inp_\emptyset$, a designated auxiliary relation $Q$ in $\prog_\alpha(\state_\emptyset)$ and $\query(\alpha(\inp_\emptyset))$ coincide. Here, $\state_\emptyset=(\inp_\emptyset, \aux_\emptyset)$, where $\aux_\emptyset$ denotes the empty auxiliary structure over the domain of $\inp_\emptyset$, and $\alpha(\inp_\emptyset)$ is the input database after applying $\alpha$. 

In this article, we are particularly interested in maintaining queries for structures of bounded treewidth. There are several ways to  adjust the dynamic setting to restricted classes $\calC$ of structures. 
Sometimes it is possible that a dynamic program itself detects that a change operation would yield a structure outside the class $\calC$.
However, here we simply disallow change sequences that construct structures outside~$\calC$. That is, in the above definition, only change sequences  $\alpha$ are considered, for which each prefix transforms an  initially empty structure into a structure from $\calC$. 
We say that a program maintains $\query$ for a class $\calC$ of structures, if $Q$ contains its result after each change sequence $\alpha$ such that the application of each prefix of $\alpha$ to $\inp_\emptyset$ yields a structure from~$\calC$. 
  
The class of queries that can be maintained by a dynamic program with first-order update formulas is called $\DynFO$. We say that a query $\query$ is in \DynFO for a class $\calC$ of structures, if there is such a dynamic program that  maintains $\query$ for $\calC$. Programs for queries in \DynFOar have three particular auxiliary relations $\leq, +, \times$ that are initialised as a linear order and the corresponding addition and multiplication relations. %

For a wide class of queries, membership in \DynFOar implies membership in \DynFO \cite{DattaKMSZ18}. 
Queries of this class have the property that the query result does not change considerably when elements are added to the domain but not to any relation.
Informally, a query $\query$ is called almost domain independent if there is a constant $c$ such that if a structure already has at least $c$ ``non-active'' elements, adding more ``non-active'' elements does not change the query result with respect to the original elements.
More formally, a query $\query$ is \emph{almost domain independent} if there is a $c \in \N$ such that, for every structure $\calA$ and every set $B \subseteq A \setminus \adom(\calA)$ with $|B| \geq c$ it holds  $\restrict{\query(\calA)}{(\adom(\calA) \cup B)} = \query(\restrict{\calA}{(\adom(\calA) \cup B)})$.

\begin{exas}\hfill
 \begin{enumerate}
  \item The binary reachability query $\Qreach$, that maps a directed input graph $G=(V,E)$ to its transitive closure relation, is almost domain independent with $c = 0$: adding any set $B \subseteq V \setminus \adom(G)$ of isolated nodes to a graph does not create or destroy paths in the remaining graph. Note that, for each node $v \in V$, the tuple $(v,v)$ is part of the query result $\Qreach(G)$, so $\Qreach(G) \neq \Qreach(\restrict{G}{\adom(G)})$ in general and therefore $\Qreach$ is not \emph{domain independent} in the sense of \cite{DattaKMSZ15}.
  \item The \FO definable Boolean query $\query_{|\neg \adom|=2}$, which is true if and only if exactly two elements are not in the active domain, is almost domain independent with $c=3$.
  \item The Boolean query $\query_{\text{even}}$, which is true for domains of even size and false otherwise, is not almost domain independent.
 \end{enumerate}
\end{exas}

Furthermore, all properties definable in monadic second-order logic are almost domain dependent.

\begin{proposition}\label{prop:msoadi}
All \MSO-definable queries are almost domain independent.
\end{proposition}

\begin{proofsketch}
 This can be easily shown by an Ehrenfeucht game. Let $\varphi$ be an \MSO-formula of quantifier depth $d$ with $e$ free (node) variables and let $c=2^dd+e$. Consider two graphs $\calA_1$ and $\calA_2$ that result from adding $c_1\ge c$ and $c_2\ge c$ isolated nodes to some graph $\calA$, respectively. Since $\varphi$ might have free variables, the game is played on $(\calA_1,\tpl a_1)$ and $(\calA_2,\tpl a_2)$, where $\tpl a_1$ and $\tpl a_2$ are tuples of elements of length $e$.
Within $d$ moves the spoiler can use at most $d$ set moves and can therefore induce at most $2^d$ different ``colours'' on $\calA_1$ and $\calA_2$. However, the duplicator can easily guarantee that, for each such ``colour'', the number of isolated nodes (which do not occur in the initial tuples on both structures) of that colour is the same in $\calA_1$ and $\calA_2$ \emph{or} it is larger than $d$, in both of them. Since the spoiler can have at most $d$ node moves, he can not make use of this difference in the two structures. 
\end{proofsketch}

\noindent The following proposition adapts Proposition 7 from \cite{DattaKMSZ18}.

\begin{proposition}\label{proposition:acisfo:almost}
Let $\query$ be an almost domain independent query. If \mbox{$\query\in \DynFOar$} then also \mbox{$\query\in \DynFO$}.
\end{proposition}

The same statement is proved for weakly domain independent queries in \cite[Proposition 7]{DattaKMSZ18}.
  A query $\query$ is weakly domain independent, if $\restrict{\query(\calA)}{\adom(\calA)} = \query(\restrict{\calA}{\adom(\calA)})$ for all structures $\calA$, that is, if it is almost domain independent with $c=0$.
The proof of \cite[Proposition 7]{DattaKMSZ18} can easily be adapted for this more general statement, so we omit a full proof here. However, for readers who are familiar with the proof of \cite[Proposition 7]{DattaKMSZ18}, we sketch the necessary changes.
  
\begin{proofsketch}
We assume familiarity with the proof of \cite[Proposition 7]{DattaKMSZ18} and only repeat its main outline.
That proof explains, for a given $\DynFOar$-program $\prog$, how to construct a $\DynFO$-program $\prog'$ that is equivalent to~$\prog$. 
This program $\prog'$ relies on the observation that a linear order, addition and multiplication can be maintained on the \emph{activated} domain \cite{Etessami98}, that is, on all elements that were part of the active domain at some time during the dynamic process. The linear order is determined by the order in which the elements are activated.

The program $\prog'$ can be regarded as the parallel composition of multiple copies, called \emph{threads}, of the same dynamic program.
Each thread simulates $\prog$ for a certain period of time:
thread $i$ \emph{starts} when $f(i-1)$ elements are activated and does some initialisation, and it is \emph{in charge} of answering the query whenever more than $f(i)$ but at most $f(i+1)$ elements are activated, for some function $f$. Thanks to weak domain independence, thread $i$ only needs to simulate $\prog$ on a domain with $f(i+1)$ elements, for which it can define the arithmetic relations $\leq$, $+$ and $\times$ based on the available arithmetic relations on the $f(i-1)$ activated elements.

For almost domain independent queries we extend this technique slightly. 
Let $c$ be the constant from almost domain independence. 
The query result for a structure with domain size $n$ coincides with the result of a simulation on a domain of size $n'$, if $n = n'$ or if there are at least $c$ non-activated elements in both domains.
To ensure that the simulation can in principle represent these $c$ elements, phase $i$ is generally in charge as long as the number of activated elements is at least $f(i)-c+1$ but at most $f(i+1)-c$, as in phase $i$ one can simulate $\prog$ on a domain of size $f(i+1)$.

However, if in the original structure the number of non-activated elements becomes smaller than $c$, the simulation has to switch to a domain of the same size as the original domain.
Therefore in phase $i$ there is not only one simulation with domain size $f(i+1)$, but one simulation (denoted by the pair $(i,\ell)$) for each domain size $\ell \in \{f(i)-c+1, \ldots, f(i+1)\}$. %
During phase $i$ the simulation denoted by $(i,f(i+1))$ is in charge unless there are fewer than $c$ non-activated elements in the original domain, which can be easily detected by a first-order formula. As soon as that happens, the simulation denoted by $(i,n)$ takes over, where $n$ is the size of the domain.
\end{proofsketch}

 \section{Algorithmic Technique} \label{section:technique}
The definition of $\DynFO$ requires that for the problem at hand each change can be handled by a first-order definable update operation. 
There are alternative definitions of \DynFO, where the initial structure is non-empty and the initial auxiliary relations can be computed within some complexity class \cite{PatnaikI97,WeberS07}. However, in a practical scenario of dynamic query answering it is conceivable that the quality of the auxiliary relations decreases over time and that they are therefore recomputed from scratch at times. We formalise this notion by a relaxed definition of maintainability in which the initial structure is non-empty, the dynamic program is allowed to apply some preprocessing, and query answers need only be given for a certain number of change steps. 

A query $\query$ is called \emph{$(\calC,f)$-maintainable}, for some complexity class\footnote{Strictly speaking $\calC$ should be a complexity class of functions. In this paper, the implied class of functions will always be clear from the stated class of decision problems.} $\calC$ and some function~\mbox{$f:\N\to\R$}, if there is a dynamic program $\prog$ and a $\calC$-algorithm $\calA$ such that for each input database $\inp$ over a domain of size $n$, each linear order $\leq$ on the domain, and each change sequence $\alpha$ of length $|\alpha| \leq f(n)$, the relation $Q$ in $\prog_\alpha(\state)$ and $\query(\alpha(\inp))$ coincide where~$\state = (\inp, \calA(\inp,\leq))$.

Although we feel that $(\calC,f)$-maintainability deserves further investigation, in this paper we exclusively use it as a tool to prove that queries are actually maintainable in $\DynFO$. To this end, we show next that every $(\AC^1,\log n)$-maintainable query is actually in $\DynFO$ and prove later that the queries in which we are interested are $(\AC^1,\log n)$-maintainable.

\begin{theorem}\label{theorem:fewerChanges}
Every  $(\AC^1,\log n)$-maintainable, almost domain independent query is in $\DynFO$.
\end{theorem}

We do not prove this theorem directly, but instead give a more general result, strengthening the correspondence between depth of the initialising circuit families and number of change steps the query has to be maintained.

\begin{theorem}\label{theorem:fewerChangesGeneral}
Let $f: \N \to \N$ be a first-order constructible function with $f \in \bigO(n)$.
Every $(\AC[f(n)],f(n))$-maintainable, almost domain independent query is in $\DynFO$.
\end{theorem}
  
\begin{proof}
Let $f: \N \to \N$ be first-order constructible with $f \in \bigO(n)$ and assume that an $\AC[f(n)]$ algorithm $\calA$ and a dynamic program $\prog$ witness that an almost domain independent query $\query$ is $(\AC[f(n)],f(n))$-maintainable. Thanks to Proposition~\ref{proposition:acisfo:almost} it suffices to construct a dynamic program $\prog'$ that witnesses $\query \in \DynFOar$. We restrict ourselves to graphs, for simplicity.

The overall idea is to use a simulation technique similar to the ones used in Proposition~\ref{proposition:acisfo:almost} and in \cite[Theorem 8.1]{SchwentickVZ18}. We first present the computations performed by $\prog'$ intuitively and later explain how they can be expressed in first-order logic.

We consider each application of one change as a \emph{time step} and refer to the graph after time step $t$ as $G_t=(V,E_t)$.
After each time step $t$, the program $\prog'$ starts a thread that is in charge of answering the query at time point $t + f(n)$.
Each thread works in two phases, each lasting $\frac{f(n)}{2}$ time steps. Roughly speaking, the first phase is in charge of simulating $\calA$ and in the second phase $\prog$ is used to apply all changes that occur from time step $t+1$ to time step $t + f(n)$. 
Using $f(n)$ many threads, $\prog'$ is able to answer the query from time point $f(n)$ onwards.

We now give more details on the two phases and describe afterwards how to deal with time points earlier than $f(n)$.
For the first phase, we make use of the equality  $\AC[f(n)]=\Ind{f(n)}$, see \cite{ImmermanDC}. Let $\psi$ be an inductive formula that is applied $d f(n)$ times, for some $d$, to get the auxiliary relations $\calA(G,\leq)$ for a given graph $G$ and the given order $\leq$.
The program $\prog'$ applies $\psi$ to $G_t$, $2d$ times during each time step, and thus the result of $\calA$ on $(G_t, \leq)$ is obtained after $\frac{f(n)}{2}$ steps. The change operations that occur during these steps are not applied to $G_t$ directly but rather stored in some additional relation. If some edge $e$ is changed multiple times, the stored change for $e$ is adjusted accordingly.

During the second phase the $\frac{f(n)}{2}$ stored change operations and the $\frac{f(n)}{2}$ change operations that happen during the next $\frac{f(n)}{2}$ steps are applied to the state after phase 1. To this end, it suffices for $\prog'$ to  apply two changes during each time step by simulating two update steps of~$\prog$. Observe that $\prog'$ processes the changes in a different order than they actually occur. However, both change sequences result in the same graph.
Since $\prog$ can maintain $\query$ for $f(n)$ changes, the program $\prog'$ can give the correct query answer for $\query$ about $G_{t+f(n)}$ at the end of phase 2, that is, at time point $t+f(n)$.

The following auxiliary relations are used by thread $i$:
\begin{itemize}
 \item a binary relation $\hat{E}_i$ that contains the edges currently considered by the thread,
 \item binary relations $\Delta^+_i$ and $\Delta^-_i$ that cache edges inserted and deleted during the first phase, respectively,
 \item a relation $\hat{R}_i$ for each auxiliary relation $R$ of $\prog$ with the same arity,
 \item and a relation $C_i$ that is used as a counter: it contains exactly one tuple which is interpreted as the counter value, according to its position in the lexicographic order induced by $\leq$.
\end{itemize} 
When thread $i$ starts its first phase at time point $t$, it sets $\hat{E}_i$ to $E_t$ and the counter $C_i$ to $0$; its other auxiliary relations are empty in the beginning. 
Whenever an edge $(u,v)$ is inserted (or deleted), $\hat{E}_i$ is not changed, $(u,v)$ is inserted into $\Delta^+_i$ (or $\Delta^-_i$)\footnote{If an edge $(u,v)$ with $(u,v) \in \Delta^-_i$ is inserted, it is instead deleted from $\Delta^-_i$, and accordingly for deletions of edges $(u,v) \in \Delta^+_i$.}, and the counter $C_i$ is incremented by one. The relations $\hat{R}_i$ are replaced by the result of applying their defining first-order formulas $2d$ times, as explained above.

When the counter value is at least $\frac{f(n)}{2}$, %
the thread enters its second phase and proceeds as follows. 
When an edge $(u,v)$ is inserted (or deleted), it applies this change and the change implied by the lexicographically smallest tuple in $\Delta^+_i$ and $\Delta^-_i$, if these relations are not empty: it simulates $\prog$ for these changes using the edge set $\hat{E}_i$ and auxiliary relations $\hat{R}_i$, replaces the auxiliary relations accordingly and adjusts $\hat{E}_i, \Delta^+_i$ and $\Delta^-_i$. Again, the counter $C_i$ is incremented.
If the counter value is $f(n)$, the thread's query result is used as the query result of $\prog'$, and the thread stops. %
All steps are easily seen to be first-order expressible.

So far we have seen how $\prog'$ can give the query answer from time step $f(n)$ onwards. For time steps earlier than $f(n)$ the approach needs to be slightly adapted as the program does not have enough time to simulate~$\calA$. The idea is that for time steps $t < f(n)$ the active domain is small and, exploiting the almost domain independence of $\query$, it suffices to compute the query result with respect to this small domain extended by $c$ isolated elements, where $c$ is the constant from almost domain independence. The result on this restricted domain can afterwards be used to define the result for the whole domain. 

Towards making this idea more precise, let $n_0, b$ be such that $bn \geq f(n)$ for all $n \geq n_0$. We focus on explaining how $\prog'$ handles structures with $n \geq n_0$, as small graphs with less than $n_0$ nodes can be dealt with separately.

The program $\prog'$ starts a new thread at time $\frac{t}{2}$ for the graph $G_{\frac{t}{2}}$ with at most $\frac{t}{2}$ edges. Such a thread is responsible for providing the query result after $t$ time steps, and works in two phases that are similar to the phases described above. It computes relative to a domain $D_t$ of size $\min\{2t+c,n\}$, where $c$ is the constant from (almost) domain independence.  The size of $D_t$ is large enough to account for possible new nodes used in edge insertions in the following $\frac{t}{2}$ change steps. The domain $D_t$ is chosen as the first $|D_t|$ elements of the full domain (with respect to $\leq$). The program $\prog'$ maintains a bijection $\pi$ between the active domain $D_{G}$ of the current graph $G$ and the first $|D_G|$ elements of the domain to allow a translation between $D_G$ and $D_t$.

 The first phase of the thread for $t$ starts at time point $\frac{t}{2}+1$ and applies $\psi$ for $8bcd$ times during each of the next time steps. This simulation of $\calA$ is finished after at most $\frac{(2t+c)b}{8bc} \leq \frac{t}{4}$ time steps, and therefore the auxiliary relations are properly initialised at time point $\frac{3t}{4}$.

In the second phase, starting at time step $\frac{3t}{4}+1$ and ending at time step $t$, the changes that occurred in the first phase are applied, two at a time. The thread is then ready to answer $\query$ at time point $t$. Since at time $t$ at most $2t$ elements are used by edges, the almost domain independence of $\query$ guarantees that the result computed by the thread relative to $D_t$ coincides with the $D_t$-restriction of the query result for $\pi(G_t)$. The query result for $G_t$ is obtained by translating the obtained result according to $\pi^{-1}$, and extending it to the full domain. More precisely, a tuple $\tpl t$ is included in the query result, if it can be generated from a tuple $\tpl t'$ of the restricted query result by replacing elements from $\pi^{-1}(D_t) \setminus \adom(G_t)$ by elements from $V \setminus  \adom(G_t)$ (under consideration of equality constraints among these elements). Again, all steps are easily seen to be first-order definable using the auxiliary relations from above.

The above presentation assumes a separate thread for each time point and each thread uses its own relations. These threads can be combined into one dynamic program as follows. Since at each time point at most $f(n)$ threads are active, we can number them in a round robin fashion with numbers $1,\ldots,f(n)$ that we can encode by tuples of constant arity. The arity of all auxiliary relations is incremented accordingly and the additional dimensions are used to indicate the number of the thread to which a tuple belongs.
\end{proof}

\section{Warm-up: 3-Colourability} \label{section:3col}

In this section, we show that the 3-colourability problem \problem{3Col} for graphs of bounded treewidth can be maintained in \DynFO. Given an undirected graph, \problem{3Col} asks whether its vertices can be coloured with three colours such that adjacent vertices have different colours. 
\begin{theorem}
  For every $k$, \problem{3Col} is in \DynFO for  graphs with treewidth at most $k$.\label{theo:3col}
\end{theorem}

The remainder of this section is dedicated to a  proof for this theorem. Thanks to Theorem~\ref{theorem:fewerChanges} and the fact that \problem{3Col} is almost domain independent, it suffices to show that \problem{3Col} is $(\AC^1,\log n)$-maintainable for graphs with treewidth at most $k$. In a nutshell, our approach can be summarised as follows. 

The $\AC^1$ initialisation computes a nice tree decomposition $(T,B)$ of width at most $4k+5$ and maximum bag size $\ell\df 4k+6$, as well as information about the 3-colourability of induced subgraphs of $G$. More precisely, it computes, for each triangle $\delta$ of~$T$ and each 3-colouring $C$ of the nodes of $B(\delta)$, whether there exists a colouring $C'$  of the inner vertices of $G(\delta)$, such that all edges involving at least one inner vertex are consistent
 with $C\cup C'$.

During the following $\log n$ change operations, the dynamic program does not need to do much. 
It only maintains a set $S$ of \emph{special} bags: for each \emph{affected} graph node $v$ that participates in any changed (i.e.~deleted or inserted) edge, $S$ contains one bag in which $v$ occurs. Also, if two bags are special, their least common ancestor is considered special and is included in $S$. 
It will be guaranteed that there are at most $4 \log n$ special bags. With the auxiliary information, a first-order formula $\varphi$ can test whether $G$ is 3-colourable as follows. By existentially quantifying $8\ell$ variables, the formula can choose two bits of information for each of the at most $4\ell\log n$ nodes in special bags. For each such node, these two bits are interpreted as encoding of one of three colours and the formula $\varphi$  checks that this colouring of the special bags can be extended to a colouring of $G$. This can be done with the help of the auxiliary relations computed during the initialisation which provide all necessary information about colourability of subgraphs induced by triangles consisting of special bags. %

Before we give a detailed proof, we need some more notation. Let $G = (V,E)$ be a graph and $(T,B)$ with $T = (I,F,r)$  a nice tree decomposition with bags of size at most $\ell$. A \emph{colouring} of a set $U$ of vertices is just a mapping from $U$ to $\{1,2,3\}$. An edge $(u, v)$ is \emph{properly coloured} if  $u$ and $v$ are mapped to different colours. For a triangle $\delta$, we say that a colouring $C$ of $B(\delta)$ is  \emph{consistent}, if there exists a colouring $C'$ of the inner vertices of $G(\delta)$ such that all edges of $G(\delta)$ are properly coloured by $C\cup C'$. Recall that $G(\delta)$ only contains edges that involve at least one inner vertex. 

We say that a tuple $\tpl v(i)=(v_1,\ldots,v_{\ell})$ \emph{represents a tree node} $i\in I$ (or, the bag $B(i)$) if $B(i)=\{v_1,\ldots,v_{\ell}\}$.
A tuple $\tpl v(\delta)=(\tpl v(i_0), \tpl v(i_1), \tpl v(i_2))$ \emph{represents the triangle} $\delta=(i_0,i_1,i_2)$.
If $\tpl v(\delta) = (v_1, \ldots, v_{3\ell})$ represents the triangle $\delta$ and $\tpl c$ is a tuple from  $\{1,2,3\}^{3\ell}$ such that $c_{j} = c_{j'}$ whenever $v_{j} = v_{j'}$ for  $j, j'\in\{1,\ldots,3\ell\}$, we write $C_{\tpl c,\tpl v}$ for the colouring  of $B(\delta)$ defined by $C_{\tpl c,\tpl v} (v_j)=c_j$, for every $j\in\{1,\ldots,3\ell\}$.

\begin{proofof}{Theorem \ref{theo:3col}}
Let $G = (V,E)$ be a graph of treewidth at most $k$. 

The $\AC^1$ initialisation first computes a $d$-nice tree decomposition $(T,B)$ with bags of size at most $\ell= 4k+6$, for the constant $d$ guaranteed to exist by Lemma \ref{prelims:decompositions:nice}, and the predecessor relation $\preceq$ of $T$. Also, it initialises the relations $\leq$ and $\BIT$. 
Next, it computes the following auxiliary relations in a bottom-up fashion with respect to $T = (I, F, r)$.  For each tuple $\tpl c\in\{1,2,3\}^{3\ell}$ the auxiliary relation $R_{\tpl c}$ contains all tuples $\tpl v(\delta)$ from $V^{3\ell}$ that represent some triangle $\delta$ such that $C_{\tpl c,\tpl v}$ is a consistent colouring of~$B(\delta)$. 

The auxiliary relations are computed inductively and bottom-up, that is, the auxiliary information for a tuple representing a triangle $(i_0, i_1, i_2)$ is computed by using the information for the triangles rooted at the two children of $i_0$. It will be easy to see that each inductive step can be defined by a first-order formula and, since $T$ has depth $d \log n$, the induction reaches a fixpoint after $d \log n$ iterations. Therefore the initialisation is in $\Ind{\log n} = \AC[\log n]$. We recall that triangles can be open, unary or proper, depending on whether they are induced by a single bag, by two, or by three bags. 

For the base case, a tuple $\tpl v$ representing an open triangle corresponding to a leaf of $T$ is in $R_{\tpl c}$ if and only if, for each $j\in\{1,\ldots\ell\}$, $c_j=c_{\ell+j}=c_{2\ell+j}$, since there are no inner vertices to worry about. 

The inductive cases are straightforward. We only describe  in detail the case of a proper triangle $\delta=(i_0,i_1,i_2)$ where $i_1$ and $i_2$ are in different subtrees of $i_0$; the other cases are similar. Let $i'_1$ and $i'_2$ be the two children of $i_0$ such that $i'_1\preceq i_1$ and $i'_2\preceq i_2$. Figure~\ref{figure:triangle-ind} illustrates this situation. By the induction hypothesis, the auxiliary information for all triangles rooted at $i'_1$ and $i'_2$ has already been computed. 
A tuple $\tpl v= \tpl v(\delta)$ is in some $R_{\tpl c}$, if there are tuples $\tpl d, \tpl e \in \{1,2,3\}^{3\ell}$ such that $\tpl u = \tpl v((i'_1,i_1,i_1)) \in R_{\tpl d}$, $\tpl w = \tpl v((i'_2,i_2,i_2)) \in R_{\tpl e}$ and it holds that
\begin{itemize}
\item $C_{\tpl c,\tpl v}$ and $C_{\tpl d,\tpl u}$ coincide on $B(i_0)\cap B(i'_1)$ and on $B(i_1)$,
\item $C_{\tpl c,\tpl v}$ and $C_{\tpl e,\tpl w}$ coincide on $B(i_0)\cap B(i'_2)$ and on $B(i_2)$, and
\item all edges over $B(i_0)\cup B(i'_1)\cup B(i'_2) \cup B(i_1) \cup B(i_2)$ of which at least one node is not in $B(i_0) \cup B(i_1) \cup B(i_2)$ are properly coloured by  $C_{\tpl c,\tpl v}\cup C_{\tpl d,\tpl u}\cup C_{\tpl e,\tpl w}$.
\end{itemize}

 \begin{figure}[t] 
    \begin{center}
    \scalebox{0.9}{
          \begin{tikzpicture}[
            xscale=1.2,
            yscale=0.7,
            font=\footnotesize,
          ]
                \node (r) at (0,5)[mnode, label=right:$r$] {};
                \node[mnode, fill=blue!50] (i0) at (0,4)[label=right:$i_0$] {};
                \node[mnode, fill=blue!50] (i1) at (-0.75,1)[label=right:$i_1$] {};
                \node[mnode, fill=blue!50] (i2) at (0.75,1)[label=right:$i_2$] {};
               \node[mnode, fill=blue!50] (i1p) at (-0.25,3)[label=left:$i'_1$] {};
                \node[mnode, fill=blue!50] (i2p) at (0.25,3)[label=right:$i'_2$] {};

                \node (a) at (-3.25,0)[invisible] {};
                \node (b) at (-2.75,0)[invisible] {};
                \node (c) at (-1.25,0)[invisible] {};
                \node (d) at (-0.25,0)[invisible] {};
               \node (m) at (0,0)[invisible] {};
                \node (e) at (0.25,0)[invisible] {};
                \node (f) at (1.25,0)[invisible] {};
                \node (g) at (2.75,0)[invisible] {};
                \node (h) at (3.25,0)[invisible] {};
      
                \draw [decoration={zigzag,segment length=4,amplitude=.9},line join=round] (r) edge[decorate]  (i0);
                \draw [line join=round] (i0) edge[decorate]  (i1p);
                \draw [line join=round] (i0) edge[decorate]  (i2p);
               \draw [decoration={zigzag,segment length=4,amplitude=.9},line join=round] (i1p) edge[decorate]  (i1);
                \draw [decoration={zigzag,segment length=4,amplitude=.9},line join=round] (i2p) edge[decorate]  (i2);

                \begin{pgfonlayer}{background}
                  \begin{scope}[very thin]
                \draw[fill=blue!10] (b.center) -- (g.center) -- (i0.center) -- cycle;
                \draw[fill = yellow!50, opacity=0.4] (b.center) -- (m.center) -- (i1p.center) -- cycle;
               \draw[fill = brown!50, opacity=0.4] (m.center) -- (g.center) -- (i2p.center) -- cycle;
                \draw[fill=white] (c.center) -- (d.center) -- (i1.center) -- cycle;
                \draw[fill=white] (e.center) -- (f.center) -- (i2.center) -- cycle;
                \draw (a.center) -- (h.center) -- (r.center) -- cycle;
              \end{scope}

               \end{pgfonlayer}
      \end{tikzpicture}

    }
    \caption{Illustration of the inductive step in the computation of colourability information for triangles in the proof of Theorem \ref{theo:3col}. }\label{figure:triangle-ind}
    \end{center}%
  \end{figure}
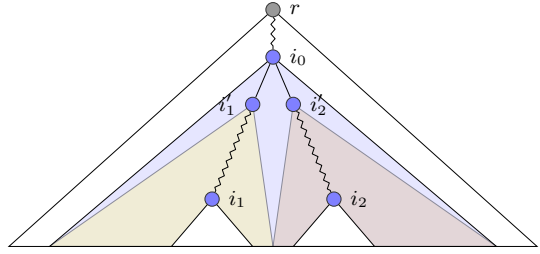

We next describe how a dynamic program can maintain 3-colourability for $\log n$ change steps starting from the above initial auxiliary relations with the help of an additional $\ell$-ary relation $S$ and another binary relation $N$. Whenever an edge $(u,v)$ is inserted into or deleted from $E$, we consider both $u$ and $v$ as \emph{affected}. With each affected graph node $v$ we associate a tree node $i(v)\in I$ such that $v\in B(i(v))$. Tree nodes of the form $i(v)$ for affected nodes $v$ are called \emph{special}. Furthermore, if node $i$ is the least common ancestor of two special nodes $i_1,i_2$ it becomes special, as well. The dynamic program keeps track of all special nodes using the relation $S$ which contains all tuples $\tpl v$ that represent some special node. Furthermore, using the relation $N$ it maintains a bijection between the first $\ell|S|$ nodes of $V$ with respect to the linear order $\leq$ and the graph nodes in $S$. We call a triangle $\delta=(i_0,i_1,i_2)$ of $T$ \emph{clean} if 
there are no special nodes in $T(\delta)$ apart from $i_0, i_1, i_2$.

It only remains to describe how a first-order formula can check 3-colourability of $G$ given the relation $S$ and the relations $R_{\tpl c}$.

We note first that within $\log n$ steps at most $2\log n$ graph nodes can be affected resulting in at most $4\log n$ special tree nodes altogether (since each new special node can contribute at most one new least common ancestor of special nodes). That is, the set $Z$ of graph nodes occurring in some tuple of $S$ contains at most $4\ell\log n$ nodes. A colouring of $Z$ can be represented by $8\ell\log n$ bits and can thus be guessed by a first-order formula by  quantifying over $8\ell$ first-order variables $x_1,\ldots,x_{4\ell},y_1,\ldots,y_{4\ell}$. More precisely, the $j$-th bits of $x_r$ and $y_r$ together represent the colour of the special node at position $(r-1)\log n+j$ with respect to the linear order represented by $N$. The first-order formula can easily check that the colouring $C$ of $S$ represented  by  $x_1,\ldots,x_{4\ell},y_1,\ldots,y_{4\ell}$ is consistent for edges between special nodes and that for each clean triangle of $T$ induced by special nodes it can be extended to a consistent colouring of the inner nodes. The latter information is available in the relations~$R_{\tpl c}$.
\end{proofof}

  \section{MSO Queries} \label{section:mso}
\newcommand{\liff}{\ensuremath{\mathop{\leftrightarrow}}\xspace}
\newcommand{\lif}{\ensuremath{\mathop{\leftarrow}}\xspace}
\newcommand{\loif}{\ensuremath{\mathop{\rightarrow}}\xspace}
\newcommand{\selem}[1]{\langle #1 \rangle}

In this section we prove a dynamic version of Courcelle's Theorem: all \MSO properties can be maintained in \DynFO for graphs with bounded treewidth. More precisely, for a given \MSO sentence $\varphi$ we consider the model checking problem $\problem{MC}_\varphi$ that asks whether a given graph $G$ satisfies $\varphi$, that is, whether $G \models \varphi$ holds.

\begin{theorem} \label{theorem:mso}
For every \MSO sentence $\varphi$ and every $k$, $\problem{MC}_\varphi$ is in \DynFO for graphs with treewidth at most $k$. 
\end{theorem}
Since, for every $k$,  guarded second-order logic (\GSO) has the same expressive power as  \MSO on graphs with treewidth at most $k$ \cite[Theorem 2.2]{Courcelle94}, we can immediately conclude the following corollary.
\begin{corollary} \label{cor:gso}
For every \GSO sentence $\varphi$ and every $k$, $\problem{MC}_\varphi$ is in \DynFO for graphs with treewidth at most $k$. 
\end{corollary}

We first give a rough sketch of the proof. Let $\varphi$ be a fixed \MSO formula of quantifier depth $d$ and $k$ a treewidth. We show that $\problem{MC}_\varphi$ is $(\AC^1,\log n)$-maintainable for graphs with treewidth at most $k$.  The construction of a dynamic program for $\problem{MC}_\varphi$ is similar to the one in the proof for \problem{3-Colourability} (Theorem~\ref{theo:3col}). At each point, the program needs to evaluate $\varphi$ on a graph $G'=(V,(E \setminus E^-) \cup E^+)$ with $n$ nodes, where $|E^- \cup E^+|=\bigO(\log n)$, using a nice tree decomposition $(T,B)$ of width $4k+5$ for the initial graph $G=(V,E)$ and auxiliary information on the \MSO type of depth $d$ for each triangle of $T$ (defined as in Section~\ref{section:preliminaries}). 

The graph $G'$ can be viewed as having  a \emph{center} $C\subseteq V$ of logarithmic size, that contains the nodes with edges in $E^- \cup E^+$ and, additionally, for each of these nodes $v$ all nodes of one bag that contains $v$. %

Furthermore, there are node sets $D_1,\ldots,D_\ell$ that, together with $C$, contain all nodes from $V$, such that the sets $D_i-C$ are pairwise disjoint and disconnected, and each set $D_i\cap C$ has size $\bigO(1)$ (cf.\ Figure \ref{figure:connection-width}). From the type information for the triangles, the program can infer the depth-$d$ \MSO types of each $G'[D_i]$.

The situation is similar as in the \emph{Composition Theorem} of Elberfeld, Grohe and Tantau \cite[Theorem 3.8]{ElberfeldGT16}. However, in their setting, the size of $C$ is bounded by a constant, and they show, very roughly, that there is a first-order formula that can be evaluated on a suitable extension of $G[C]$ by information on the types of the $G[D_i]$ to yield the same result as $\varphi$ on~$G$.

We show that in our setting one can construct an \MSO formula $\psi$  such that $G'\models \varphi$ if and only if $\calB \models \psi$, where $\calB$ is a structure, which extends $G'[C]$ by type information about the $G'[D_i]$.  This construction is detailed in Section~\ref{subsction:fv} below. It uses well-known techniques and, in particular, a composition theorem by Shelah (cf.\ Theorem~\ref{theorem:shelah}). 

The dynamic program then uses a first-order formula (to be evaluated in a suitable extension of $G'$ with auxiliary relations) that is obtained from $\psi$ by replacing the second-order quantification over $C$ by first-order quantification over $V$. This is possible, since sets of size $\bigO(\log n)$ over $C$ can be encoded by $\bigO(1)$ elements of $V$.
 
In the remainder of this section we make these ideas more precise. In the next subsection we state and prove a composition theorem for graphs with a center of the form described above. Then, in Section~\ref{subsction:msorest}, we show how this theorem is applied to dynamically evaluate an \MSO formula $\varphi$.

\subsection{A Feferman–Vaught-type composition theorem}\label{subsction:fv}
In the following, we give an adaptation of the Feferman–Vaught-type composition theorem from \cite{ElberfeldGT16} that will be useful for maintaining \MSO properties. 
Intuitively, the idea is very easy, but the formal presentation will come with some technicalities. For ease of presentation, we explain the basic idea for graphs first.

We consider graphs $G=(V,E)$ with a \emph{center} $C\subseteq V$, such that there are sets $D_1,\ldots,D_\ell$ such that, for some $w>0$, the following conditions hold.
\begin{itemize}
\item $C\cup \bigcup_{i=1}^\ell  D_i= V$.
\item For all $i\not=j$, $D_i\cap D_j \subseteq C$.
\item All edges in $E$ have both end nodes in $C$ or in some $D_i$. 
\item For every $i$, $|D_i\cap C|\le w$.
\item For each $i$ there is some element $v_i\in D_i\cap C$ that is not contained in any $D_j$, for $j\not = i$.
\end{itemize}
In this case, we say that $C$ has \emph{connection width} $w$ in $G$. See Figure \ref{figure:connection-width} for an illustration.

  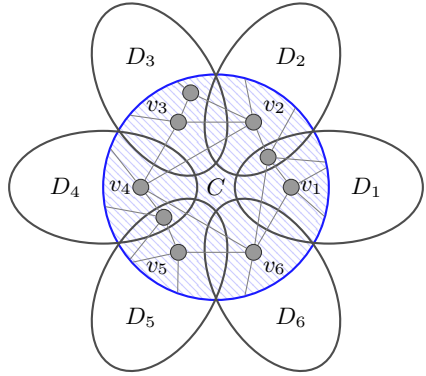
\begin{figure}[t] 
    \centering
          \begin{tikzpicture}[
            font=\footnotesize,
            exedges/.style={gray}
          ]
          \node [draw, blue!90, circle, thick, minimum height = 3cm, pattern=north west lines, pattern color=blue!20, text=black] at (0,0) {$C$};
          \foreach \x/\y in {1/0,2/60,3/120,4/180,5/240,6/300}
          { 
          \node [draw, ellipse, thick, black!70, rotate=\y, minimum width=2.5cm, minimum height=1.5cm] at   (\y:1.5cm)   {};
          \node [] at   (\y:2cm)   {$D_\x$};
          \node[mnode, label={[label distance=-0.13cm]\y:$v_\x$}] (v\x) at (\y:1cm) {}; 
		  }
		  \node[mnode] (v12) at (30:0.8cm) {};
		  \node[mnode] (v33) at (105:1.3cm) {};
		  \node[mnode] (v45) at (210:0.8cm) {};
		  \draw[exedges] (v1) edge (v12) edge (v6); 
		  \draw[exedges] (v12) edge (v6);
		  \draw[exedges] (v2) edge (v12) edge (v33);
		  \draw[exedges] (v4) edge (v2) edge (v45) edge (v6); 
		  \draw[exedges] (v5) edge (v45) edge (v6); 
		  \draw[exedges] (v3) edge (v4) edge (v33) edge (v2);
		  \draw[exedges] (v1) edge (10:1.49cm) edge (-15:1.49cm);
		  \draw[exedges] (v12) edge (13:1.49cm) edge (45:1.49cm);
		  \draw[exedges] (v2) edge (75:1.49cm);
		  \draw[exedges] (v3) edge (140:1.49cm);
		  \draw[exedges] (v33) edge (110:1.49cm);
		  \draw[exedges] (v4) edge (160:1.49cm);
		  \draw[exedges] (v45) edge (190:1.49cm) edge (220:1.49cm);
		  \draw[exedges] (v5) edge (222:1.49cm) edge (250:1.49cm);
		  \draw[exedges] (v6) edge (285:1.49cm);
      \end{tikzpicture}
    \caption{Sketch of a graph with center $C$ of connection width $2$, highlighted in blue, and petals $D_1, \ldots, D_6$.}\label{figure:connection-width}
  \end{figure}

We refer to the sets $D_1,\ldots,D_\ell$ as \emph{petals} and the nodes $v_1,\ldots,v_\ell$ as \emph{identifiers} of their respective petals. We emphasise that $\ell$ is bounded by $|C|$, but not assumed to be bounded by a constant.
Readers who have read the proof of Theorem~\ref{theo:3col} can roughly think of $C$ as the set of vertices from special bags. %

Our goal is to show the following. 
If a graph $G$ has a center $C$ of connection width $w$, for come constant $w$, then $\restrict{G}{C}$ can be extended by the information about the \MSO types of its petals in a suitable way, resulting in a structure $\calB$ with universe $C$, such that \MSO formulas over $G$ have equivalent \MSO formulas over $\calB$.

In the following, we work out the above plan in more detail. Although, for Theorem~\ref{theorem:mso}, we need the composition theorem only for coloured graphs with some constants, we deal in the following with arbitrary signatures.
We fix some relational signature $\Sigma$ and assume that it contains a unary relation symbol $C$.

The definition of the connection width of sets $C$ easily carries over to $\Sigma$-structures. In particular, tuples need to be entirely in $C$ or in some petal $D_i$, and  all constants of the structure need to be included in $C$.
For every $i$, we call the set $I_i\df D_i\cap C$ the \emph{interface} of $D_i$ and the nodes of $D_i-C$  \emph{inner elements} of~$D_i$.

Let  $\calA$ be a $\Sigma$-structure, $C$ a center of connection width $w$ with petals $D_1, \ldots, D_\ell$.
For every $i$, let $\tpl u^i=(u^i_1,\ldots,u^i_w)$ be a tuple of elements from the interface $I_i$ of $D_i$ such that $u^i_1$ is an identifier of its petal $D_i$ and every node from $I_i$ occurs in $\tpl u^i$. By $(\calA_i, \tpl u^i)$ we denote the substructure of $\calA$ induced by  $D_i$  with $u^i_1,\ldots,u^i_w$ as constants but \emph{without} all tuples over~$C$, i.e., $(\calA_i, \tpl u^i)$ only contains tuples with at least one inner element of~$D_i$.

Let $d>0$. The \emph{depth $d$, width $w$ \MSO indicator structure} of $\calA$ relative to $C$ and tuples $\tpl u^i$ is the unique structure $\calB$ which expands  $\restrict{\calA}{C}$ by the following relations:
\begin{itemize}
\item a $w$-ary relation $J$ that contains all tuples $\tpl u^i$, and
\item for every depth-$d$ \MSO-type $\tau$ over $\Sigma\cup\{c_1,\ldots,c_w\}$, a unary relation $R_\tau$ containing those identifier nodes $u^i_1$ for which the depth-$d$ \MSO-type of $(\calA_i,\tpl u^i)$ is $\tau$.  
\end{itemize}

We note that different choices of the tuples $\tpl u^i$ result in different indicator structures and we denote the set of all indicator structures of $\calA$ relative to $C$ by~$\calS(\calA,C,w,d)$.

We are now ready to formulate the desired composition theorem.

\begin{theorem}\label{theorem:center}
  For each $d>0$, every \MSO sentence $\varphi$ with depth $d$, and each $w$, there is a number $d'$ and a \MSO sentence $\psi$ such that for every $\Sigma$-structure $\calA$, every center $C$ of $\calA$ with connection width $w$ and every $\calB\in \calS(\calA,C,w,d')$ it holds
    $\calA \models \varphi$  if and only if $\calB\models \psi$.
\end{theorem}

The proof of Theorem \ref{theorem:center} uses Shelah's \emph{generalised sums} \cite{Shelah75}. We follow the exposition from Blumensath et al.~\cite{BlumensathCL08}.
In a nutshell, a generalised sum is a composition of several disjoint \emph{component structures} along an \emph{index structure}.
Shelah's composition theorem states that \MSO sentences on a generalised sum can be translated to \MSO sentences on the index structure enriched by \MSO type information on the components.

We apply the composition theorem on the basis of the following ideas, illustrated in Figure \ref{fig:shelah}.
From the center $C$ and the petals $D_i$ we define an index structure $\calI$ and component structures $\calD_i$, respectively. 
In the generalised sum, these disjoint structures are again composed into a structure that is very similar to $\calA$. More precisely, $\calA$ can be defined in the generalised sum by a first-order interpretation, and thus, thanks to Lemma~\ref{lemma:fointerpreting} below, we can translate the formula $\varphi$ for $\calA$ into a formula $\varphi'$ on the generalised sum.
Shelah's composition theorem then provides a translation of $\varphi'$ into an $\MSO$ formula $\psi'$ on the structure $\calI$ enriched with type information on the structures $\calD_i$. 
This enriched index structure is again very similar to an \MSO indicator structure $\calB$: there is a first-order interpretation that defines the enriched index structure in $\calB$. As a consequence, by  Lemma~\ref{lemma:fointerpreting} again, the formula $\psi'$ can be translated into a formula $\varphi$ on $\calB$.

Before we proceed to the proof of Theorem \ref{theorem:center}, we formally introduce the notions of a generalised sum and a first-order interpretation, and state the corresponding results on translations of \MSO formulas.
We start with first-order interpretations.

\begin{definition}
Let $\Sigma, \Gamma$ be relational signatures. A \emph{first-order interpretation} $\Upsilon$ from $\Sigma$ to $\Gamma$ consists of a first-order formula $\varphi_U(x)$ and first-order formulas $\varphi_R(x_1, \ldots, x_r)$ for each $r$-ary relation symbol $R \in \Gamma$, each over signature $\Sigma$.

The first-order interpretation $\Upsilon$ \emph{interprets}, in a $\Sigma$-structure $\calA$, the $\Gamma$-structure $\Upsilon(\calA)$ with universe $U^{\Upsilon(\calA)} \df \{a \mid \calA \models \varphi_U(a)\}$ and relations \[R^{\Upsilon(\calA)} \df \{(a_1, \ldots, a_r) \mid \calA \models \varphi_R(a_1, \ldots, a_r), a_1, \ldots, a_r \in U^{\Upsilon(\calA)} \}\] for each $R \in \Gamma$.
\end{definition}

A first-order interpretation from $\Sigma$ to $\Gamma$ not only interprets a $\Gamma$-structure in a $\Sigma$-structure, it also translates $\Gamma$-formulas to $\Sigma$-formulas.

\begin{lem}[{see e.g.~\cite[Proposition 3.2]{BlumensathCL08}}] \label{lemma:fointerpreting}
Let $\Upsilon$ be a first-order interpretation from $\Sigma$ to $\Gamma$. For every \FO (\MSO) formula $\varphi(x_1, \ldots, x_\ell)$ over $\Gamma$ there is an \FO (\MSO) formula $\varphi^\Upsilon(x_1,\ldots, x_\ell)$ over $\Sigma$ such that $\calA \models \varphi^\Upsilon(a_1,\ldots, a_\ell) \Leftrightarrow \Upsilon(\calA) \models \varphi(a_1, \ldots, a_\ell)$ for all $\Sigma$-structures $\calA$ and all elements $a_i \in U^{\Upsilon(\calA)}$.
\end{lem}

We now turn to the definition of generalised sums.

\begin{definition}[\cite{Shelah75}, formulation following \cite{BlumensathCL08}]\label{def:generalsum}
Let $\calI = (I,S_1,\ldots, S_r)$ be a structure and $(\calD_i)_{i \in I}$ a sequence of structures $\calD_i = (D_i, R^i_1, \ldots, R^i_t)$ indexed by elements $i$ of $\calI$.
The \emph{generalised sum} of $(\calD_i)_{i \in I}$ is the structure \[ \sum_{i \in I} \calD_i \df (U, \sim, R'_1, \ldots, R'_t, S'_1, \ldots, S'_r)\]
with universe $U \df \{\selem{i,a} \mid i \in I, a \in D_i\}$ and relations
\begin{itemize}
 \item $\selem{i,a} \sim \selem{i',a'}$ if and only if $i = i'$
 \item $R'_j \df \{( \selem{i,a_1},\ldots, \selem{i,a_\ell} ) \mid i \in I, (a_1, \ldots, a_\ell) \in R^i_j\}$
 \item $S'_j \df \{( \selem{i_1,a_1},\ldots,\selem{i_\ell,a_\ell} ) \mid (i_1, \ldots, i_\ell) \in S_j, a_k \in D_{i_k}$ for all $k \in \{1, \ldots, \ell\} \}$
\end{itemize}
\end{definition}

The structures $\calI$ and $\calD_i$ in this definition are also referred to as \emph{index structure} and \emph{component structures}, respectively. 

\begin{theorem}[{\cite{Shelah75}, formulation from \cite[Theorem 3.16]{BlumensathCL08}}] \label{theorem:shelah}
From every \MSO sentence~$\varphi$, a finite sequence $\chi_1, \ldots, \chi_s$ of \MSO formulas and an \MSO formula $\psi$ can be constructed such that
 \[ \sum_{i \in I} \calD_i \models \varphi \: \Leftrightarrow \: (\calI, \llbracket \chi_1 \rrbracket, \ldots, \llbracket \chi_s \rrbracket) \models \psi\]
for all index structures $\calI$ and component structures $\calD_i$, where $\llbracket \chi \rrbracket \df \{i \in I \mid \calD_i \models \chi \}$.  
\end{theorem}
Intuitively, the formulas $\chi_i$ encode the \MSO type information on the component structures of the generalized sum.

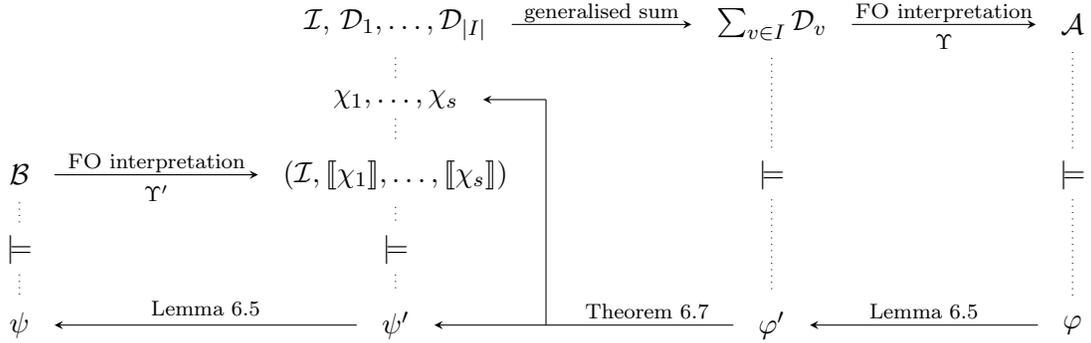
\begin{figure}[t]
\begin{tikzpicture}
 \node (1) at (4,0) {$\calI$, $\calD_1, \ldots, \calD_{|I|}$};
 \node (2) at (9,0) {$\sum_{v \in I} \calD_v$};
 \node (3) at (13,0) {$\calA$};
 \node (4) at (-1,-2) {$\calB$};
 \node (5) at (4,-2) {$(\calI, \llbracket \chi_1 \rrbracket, \ldots, \llbracket \chi_s \rrbracket)$};
 \node (9) at (13,-4) {$\varphi$};
 \node (8) at (9,-4) {$\varphi'$};
 \node (7) at (4,-4) {$\psi'$};
 \coordinate (x) at (6,-4);
 \coordinate (y) at (6,-1);
 \node (7b) at (4,-1) {$\chi_1, \ldots, \chi_s$};
 \node (6) at (-1,-4) {$\psi$};
 \node (m1) at (-1,-3) {$\models$};
 \node (m2) at (4,-3) {$\models$};
 \node (m3) at (9,-2) {$\models$};
 \node (m4) at (13,-2) {$\models$};
 \draw[shorten <=0.1cm] (3) edge[dotted] (m4);
 \draw[shorten >=0.1cm] (m4) edge[dotted] (9);
 \draw[shorten <=0.1cm] (2) edge[dotted] (m3);
 \draw[shorten >=0.1cm] (m3) edge[dotted] (8);
 \draw[shorten <=0.1cm] (4) edge[dotted] (m1);
 \draw[shorten >=0.1cm] (m2) edge[dotted] (7);
 \draw[shorten <=0.1cm] (5) edge[dotted] (m2);
 \draw[shorten >=0.1cm] (m1) edge[dotted] (6);
 \draw[shorten <=0.1cm] (1) edge[dotted] (7b);
 \draw[shorten >=0.1cm] (7b) edge[dotted] (5);
 \draw[shorten >=0.2cm, shorten <=0.2cm] (1) edge[-stealth] node[above=-0.1cm] {\scriptsize generalised sum}(2);
 \draw[shorten >=0.2cm, shorten <=0.2cm] (2) edge[-stealth] node[above=-0.1cm] {\scriptsize \FO interpretation} node[below=-0.0cm] {\scriptsize $\Upsilon$}(3);
 \draw[shorten >=0.2cm, shorten <=0.2cm] (4) edge[-stealth] node[above=-0.1cm] {\scriptsize \FO interpretation} node[below=-0.0cm] {\scriptsize $\Upsilon'$}(5);
 \draw[shorten >=0.2cm, shorten <=0.2cm] (9) edge[-stealth] node[above=-0.05cm] {\scriptsize Lemma \ref{lemma:fointerpreting}}(8);
 \draw[shorten >=0cm, shorten <=0.2cm] (8) edge node[above=-0.05cm] {\scriptsize Theorem \ref{theorem:shelah}} (x);
 \draw[shorten >=0.2cm, shorten <=0cm] (x) edge[-stealth] (7);
 \draw[] (x) edge[] (y);
 \draw[shorten >=0.2cm, shorten <=0cm] (y) edge[-stealth] (7b);
 \draw[shorten >=0.2cm, shorten <=0.2cm] (7) edge[-stealth] node[above=-0.05] {\scriptsize Lemma \ref{lemma:fointerpreting}} (6); 
\end{tikzpicture}
\caption{Overview of the proof strategy of Theorem \ref{theorem:center}.}\label{fig:shelah}
\end{figure}

With the necessary notions in place, we can now prove Theorem \ref{theorem:center}.

\begin{proofof}{Theorem \ref{theorem:center}}
Suppose that $\calA$ is a $\Sigma$-structure with center $C$ of connection width~$w$, tuples $\tpl u^i$ collecting the interface nodes for each petal $D_i$ as described above, and let $\varphi$ be an \MSO formula over signature $\Sigma$. The proof is in three steps, depicted in Figure \ref{fig:shelah}:
\begin{enumerate}[(A)]
  \item We present an index structure $\calI$ and component structures $\calD_i$, and show that there is an FO-interpretation $\Upsilon$ that interprets the structure $\calA$ in the generalised sum $\sum_{i \in I} \calD_i$.  Thus there is an \MSO formula $\varphi'$ such that $\calA \models \varphi$ if and only if $\sum_{i \in I} \calD_i \models \varphi'$ by Lemma~\ref{lemma:fointerpreting}.  
  \item  From Theorem \ref{theorem:shelah} we obtain formulas $\psi'$ and $\chi_1, \ldots, \chi_s$ such that $\sum_{i \in I} \calD_i \models \varphi'$ if and only if $(\calI, \llbracket \chi_1 \rrbracket, \ldots, \llbracket \chi_s \rrbracket) \models \psi'$.
  \item Then we show that there is an FO-interpretation $\Upsilon'$ that interprets $(\calI, \llbracket \chi_1 \rrbracket, \ldots, \llbracket \chi_s \rrbracket)$ in each \MSO indicator structure $\calB \in \calS(\calA, C, w, d')$ with appropriate $d'$. Thus there is an MSO formula $\psi$ that satisfies $\calB \models \psi$ if and only if $(\calI, \llbracket \chi_1 \rrbracket, \ldots, \llbracket \chi_s \rrbracket) \models \psi'$ by Lemma~\ref{lemma:fointerpreting}.
\end{enumerate}
Combining these three steps allows us to conclude 

\[\calA \models \varphi \quad\stackrel{(A)}{\Longleftrightarrow}\quad \sum_{i \in I} \calD_i \models \varphi'  \quad\stackrel{(B)}{\Longleftrightarrow}\quad (\calI, \llbracket \chi_1 \rrbracket, \ldots, \llbracket \chi_s \rrbracket) \models \psi' \quad\stackrel{(C)}{\Longleftrightarrow}\quad \calB \models \psi\]
for any  \MSO indicator structure $\calB \in \calS(\calA, C, w, d')$.

It remains to prove (A) and (C), since (B) is a direct application of  Theorem \ref{theorem:shelah}. 

Towards proving (A) we construct structures $\calI$ and $\calD_i$ which are closely related to the substructure $\restrict{\calA}{C}$ of $\calA$ and the structures $(\calA_i, \tpl u^i)$, respectively. Let $R_1,\ldots,R_q$ be the relation symbols of $\Sigma$ and let  $\Gamma=\{S_1,\ldots,S_q\}$.
The structure $\calI$ has universe $C$, $\Gamma$-relations $S_ i$ defined as the restriction of the respective $\Sigma$-relation $R_i$ of $\calA$ to $C$, and the relation $J$ as described above. Each structure $\calD_v$ for $v \in C$ is over signature $\Sigma \cup \{U_1 \ldots, U_w \}$ and defined as follows. If $v$ is an identifier $u^i_1$ of a petal, then $\calD_v = (\calA_i, \{u^i_1\}, \ldots, \{u^i_w\})$, that is, the restriction of $\calA$ to the elements from $D_i$, without any tuples consisting only of elements from $C  \cap D_i$, and with additional unary, singleton relations $U_1, \ldots, U_w$ such that $U_j$ includes only the $j$-th interface node $u^i_j$. If $v$ is no identifier node then $\calD_v$ is the structure with universe $\{v\}$ and empty relations.

In the generalised sum $\calS \df \sum_{v \in I} \calD_v= (U, \sim, R'_1, \ldots, R'_q, U'_1,\ldots U'_w,S'_1, \ldots, S'_q,J')$, the universe $U$ consists of elements of the form $\selem{u^i_1,w}$, where $w\in D_i$, and of the form $\selem{u,u}$ where $u\in C$ is not an identifier of any petal. We emphasise that the formulas of the interpretation that defines (a copy of) $\calA$ in $\calS$ can \emph{not} access the components $u$ and $v$ of an element  $\selem{u,v}\in U$. 
However, $U$ can be partitioned into  four kinds of elements, each of which can easily be distinguished from the others in a first-order fashion:
\begin{enumerate}[(i)]
\item Elements of the form $\selem{v,v}$, for which $D_v$ is not a petal. They can be identified since they constitute an equivalence class of size 1 with respect to $\sim$; 
\item Elements of the form $\selem{v,v}$, for which $D_v$ is a petal. These are precisely the elements in $U'_1$;
\item  Elements of the form $\selem{v,u}$, where $u\in C$ is in the petal $D_v$. These elements occur in some $U'_i$, for $i>1$ (but not in $U'_1$);
\item Elements of the form $\selem{v,u}$, where $u\not\in C$ is in the petal $D_v$. They do not occur in any $U'_i$ and are not of type (i). 
\end{enumerate}
Elements of type (iv) are in one-to-one correspondence with the  inner elements of petals $D_i$. Elements from $C$ might have several copies in $U$, but only one of the types (i) or~(ii). Thus, the formula that defines the universe for the first-order interpretation $\Upsilon$ of $\calA$ in $\sum_{i \in I} \calD_i$ simply drops all elements of type (iii). Tuples of $\calA$ of a relation $R_i$ that  entirely consist of nodes from $C$ (that is, elements of type (i) or (ii) in $\calS$) are directly induced by the corresponding relation $S'_i$.

In order to define tuples with at least one node of type (iv) in $\calA$, we first observe that it can be expressed in a first order fashion, whether for a type (iii) element $\selem{v_1,u_1}$ and a type (i) element $\selem{v_2,v_2}$ it holds $u_1=v_2$, i.e., that, intuitively, $\selem{v_2,v_2}$ is the copy representing $u_1$ in $\calS$ that survives in the universe of the interpretation.
We claim that this condition holds, if and only if $\selem{v_2,v_2}$ occurs as the $i$-th entry in some tuple of $J'$ with  first entry $\selem{v_1,u_1}$, where $i$ is the unique number such that $\selem{v_1,u_1} \in U'_i$. From this claim, first-order expressibility follows instantly.
The ``only if''-part of the claim is straightforward. For the ``if''-part, it follows  from the latter condition that there is a tuple with first entry $v_1$ and $i$-th entry $v_2$ in $J$, by the definition of $J'$. Since $\selem{v_1,u_1} \in U'_i$, there is also a tuple in $J$ with $v_1$ as first entry and $u_1$ as $i$-th entry. However, since $J$ has at most one tuple with any given value as first entry, $u_1=v_2$ follows, as claimed.

A tuple with some element $\selem{v,u}$ of type (iv) is now in a relation $R_i$ of the interpretation, if it can be transformed into a tuple of $R'_i$ by replacing some elements $\selem{w,w}$ of type (i)  with $\selem{v,w}$.

It follows from the construction that $\Upsilon(\calS)$ is isomorphic to $\calA$ and therefore $\calA \models \varphi$ if and only if $\Upsilon(\calS) \models \varphi$. We obtain the formulas $\varphi'$, $\chi_1, \ldots, \chi_s$ and $\psi'$ as explained above. 
Let $d'$ be the maximal quantifier depth of any formula $\chi_j$.  This concludes step (A) of the proof.

Towards proving (C), recall that we need to show that there is a first-order interpretation $\Upsilon'$ which interprets $(\calI, \llbracket \chi_1 \rrbracket, \ldots, \llbracket \chi_s \rrbracket)$ in $\calB$, for any $\calB \in \calS(\calA, C, w, d')$. Let $\calB$ be such a structure. The universe of $\calI$ is $C$, that is, the same as the universe of $\calB$. Thus the formula of $\Upsilon'$ that defines the universe is trivial.

For the definitions of the relations $\llbracket \chi_j \rrbracket$, the idea is as follows.
If $v$ is \emph{not} an identifier $u^i_1$ of a petal, then $v \in  \llbracket \chi_j \rrbracket$ if and only if $\chi_j$ holds in the structure consisting of only one element and with empty relations, which can be hard-coded in the defining formula.
Otherwise, if $v = u^i_1$ for some $i$, we need to determine whether $\chi_j$ holds in the structure $\calD_v = (\calA_i, \{u^i_1\}, \ldots, \{u^i_w\})$, which is a structure over signature $\Sigma \cup \{U_1, \ldots, U_w\}$.
The structure $\calB$ contains information about the \MSO types of the structures $(\calA_i, \tpl u^i)$, but $(\calA_i, \tpl u^i)$ is a structure over signature $\Sigma \cup \{c_1, \ldots, c_w\}$.
Yet it is easy to see that for the formula $\chi'_j$ that is obtained from $\chi_j$ by replacing every atom $U_k(x)$ by $x = c_k$ it holds $(\calA_i, \tpl u^i) \models \chi'_j$ if and only if $\calD_v \models \chi_j$.
So, in this case $v \in  \llbracket \chi_j \rrbracket$ if and only if $v \in R_\tau^\calB$ for a depth-$d'$ \MSO type $\tau$ with $\chi'_j \in \tau$.
All these conditions can even be expressed by quantifier-free first-order formulas for fixed formulas $\chi_j$. 
As a result, by Lemma \ref{lemma:fointerpreting} we obtain from $\Upsilon'$ and $\psi'$ a formula $\psi$ with $\calB \models \psi \Leftrightarrow \calA \models \varphi$.
\end{proofof}

\subsection{The dynamic program}\label{subsction:msorest}

We proceed to show that every \MSO-definable property can be maintained in $\DynFO$, and thus prove Theorem \ref{theorem:mso}. Thanks to Theorem~\ref{theorem:fewerChanges} and Proposition~\ref{prop:msoadi} it suffices to show  that $\problem{MC}_\varphi$ is $(\AC^1,\log n)$-maintainable for graphs $G$ with treewidth at most $k$.  

The idea for our dynamic program is similar to the idea for maintaining 3-colourability: during its initialisation the program constructs a tree decomposition and appropriate \MSO types for all triangles (instead of partial colourings as in the proof of Theorem \ref{theo:3col}). During the change sequence, a set $C$ of nodes is defined that contains, for each affected graph node $v$, all nodes of at least one special bag containing $v$. The set $C$ has connection width $w$ for some constant $w$ and the dynamic program basically maintains an \MSO indicator structure $\calB$ for $G$ relative to~$C$. As there are only $\log n$ many change steps, the size of $C$ is bounded by $\bigO(\log n)$.

By Theorem \ref{theorem:center} there is an \MSO formula $\psi$ with the property that $G \models \varphi$ if and only if $\calB \models \psi$.
Although the dynamic program maintains $\calB$, it cannot directly evaluate $\psi$, as it is restricted to use first-order formulas. 
For this reason we first show that second-order quantification over sets of size $\bigO(\log n)$ can be simulated in first-order logic, if a particular relation is present. Afterwards we present the details of the dynamic program.

We call an \MSO-formula \emph{$C$-restricted}, if all its quantified subformulas are of one of the following forms.
\begin{itemize}
\item $\exists x\; (C(x) \land \varphi)$ or $\forall x\; (C(x) \loif \varphi)$,
\item $\exists X\; (\forall x (X(x) \loif C(x)) \land \varphi)$ or $\forall X\; (\forall x (X(x) \loif C(x)) \loif \varphi)$.
\end{itemize}

Let $\calA$ be a structure with a unary relation $C$ and a $(k+1)$-ary relation $\text{Sub}$, for some $k$.  
We say that $\text{Sub}$ \emph{encodes subsets of $C$} if, for each subset $C'\subseteq C$, there is a $k$-tuple $\tpl t$ such that, for every element $c\in C$ it holds $c\in C'$ if and only if $(\tpl t,c)\in \text{Sub}$. Clearly, such an encoding of subsets only exists if $|V|^k\ge 2^{|C|}$ and thus if $|C|\le k\log |V|$.

\begin{proposition}\label{prop:log}
  For each $C$-restricted \MSO-sentence $\psi$ over a signature $\Sigma$ (containing $C$) and every $k$ there is a first-order sentence $\chi$ over $\Sigma \cup \{S\}$ where $S$ is a $(k+1)$-ary relation symbol such that, for every $\Sigma$-structure $\calA$ and $(k+1)$-ary relation $\text{Sub}$ that encodes subsets of $C$ (of $\calA$), it holds 
    $\calA \models \psi \text{ if and only if } (\calA, \text{Sub}) \models \chi$.
\end{proposition}
\begin{proof}%
  The proof is straightforward. Formulas $\exists X\; (\forall x (X(x) \loif C(x)) \land \varphi)$ are translated into formulas $\exists \tpl x\; \varphi'$, where $\tpl x$ is a tuple of $k$ variables and $\varphi'$ results from $\varphi$ by simply replacing every atomic formula $X(y)$ by $\text{Sub}(\tpl x,y)$. Universal set quantification is translated analogously. 
\end{proof}

\begin{proofof}{Theorem \ref{theorem:mso}}
Thanks to Theorem~\ref{theorem:fewerChanges} and Proposition~\ref{prop:msoadi}  it suffices to show that $\problem{MC}_\varphi$ is $(\AC^1,\log n)$-maintainable in \DynFO for graphs with treewidth at most $k$. Let $d$ be the quantifier depth of $\varphi$ and let $d'$ and $\psi$ be the number and the \MSO sentence guaranteed to exist by Theorem \ref{theorem:center}. 

Given a graph $G = (V,E)$, the $\AC^1$ initialisation first ensures that relations $\leq, +, \times$ and $\BIT$ are available.
Then it computes a $d_{\text{tree}}$-nice tree decomposition $(T,B)$ with $T = (I,F,r)$ with bags of size at most $\ell \df 4k+6$, for the constant $d_{\text{tree}}$ guaranteed to exist by Lemma~\ref{prelims:decompositions:nice}, together with the predecessor order $\preceq$ on $I$. 
With each node $i \in I$, we associate a tuple $\tpl v(i)=(v_1,\ldots,v_m,v_1,\ldots,v_1)$ of length $\ell$, where $B(i)=\{v_1,\ldots,v_m\}$ and $v_1<\cdots<v_m$. That is, if the bag size of $i$ is $\ell$, this tuple just contains all graph nodes of the bag in increasing order. If the bag size is smaller, the smallest graph node is repeated. 
Similarly, with each triangle $\delta = (i_0,i_1,i_2)$ such that the subgraph $G(\delta)$ has at least one inner node, we associate a tuple $\tpl v(\delta) = (v(\delta),\tpl v(i_0),\tpl v(i_1),\tpl v(i_2))$, where $v(\delta)$ denotes the smallest inner node of $G(\delta)$ with respect to $\leq$.

The dynamic program further uses auxiliary relations $S, C, N$, and $D_\tau$, for each depth-$d'$ \MSO type $\tau$ over the signature that consists of the binary relation symbol $E$ and $3 \ell + 1$ constant symbols $c_1, \ldots, c_{3 \ell +1}$. 
The intended meaning is that $C$ is a center of $G$ with connection width $3 \ell +1$ and that from these relations an \MSO indicator structure $\calB$ relative to $C$ can be defined in first-order. 

The relation $S$ stores tuples $\tpl v(i)$ representing special bags, as in the proof of Theorem~\ref{theo:3col}.  The relations $D_\tau$ provide \MSO type information for all triangles. More precisely, for each triangle $\delta = (i_0, i_1, i_2)$ for which the subgraph $G(\delta)$ has at least one inner node, $D_\tau$ contains the tuple $\tpl v(\delta)$ if  and only if the \MSO depth-$d'$ type of $(G(\delta),\tpl v(\delta))$ is $\tau$.

The set $C$ always contains all graph nodes that occur in special bags (and thus in $S$), plus one inner node $v(\delta)$, for each maximal clean triangle\footnote{For such a triangle $\delta=(i_0,i_1,i_2)$, the nodes $i_0, i_1, i_2$ are exactly the special nodes in $T(\delta)$.} $\delta$ with at least one inner node. 
The relation $N$ defines a bijection between $C$ and an initial segment of~$\leq$. 

We observe that $C$ is a center of $G$ and that the petals induced by $C$ correspond to the maximal clean triangles, with respect to the special nodes stored in $S$, with at least two inner nodes. 
The interface $I(\delta)$ of a petal corresponding to a maximal clean triangle $\delta = (i_0, i_1, i_2)$ contains the nodes from $B(i_0), B(i_1),$ and $B(i_2)$ as well as the node $v(\delta)$, so $C$ has connection width $w\df 3\ell+1$.
Figure \ref{figure:decomposition-and-triangles} gives  an illustration.

  \begin{figure}[t] 
    \centering
\begin{tikzpicture}[->,>=stealth',level/.style={sibling distance = 6cm/#1,
  level distance = 1.5cm},
  special/.style = {circle, draw, fill=blue!90},
  normal/.style = {circle, draw, fill=white}] 
\node [special] (r) {}
    child{ node [special] (n0) {} 
            child{ node [normal] (n00) {} 
            		child{ node [normal] (n000) {}} 
				child{ node [special] (n001) {}}
            }
            child{ node [special] (n01) {}
				child{ node [normal] (n010) {}}
				child{ node [normal] (n011) {}}
            }                            
    }
    child{ node [normal] (n1) {}
            child{ node [normal] (n10) {} 
				child{ node [normal] (n100) {}}
				child{ node [normal] (n101) {}}
            }
            child{ node [special] (n11) {}
				child{ node [normal] (n110) {}}
				child{ node [normal] (n111) {}}
            }
		}
; 
\begin{pgfonlayer}{background}
\filldraw[fill = blue!30, opacity=0.4] ([yshift = 0.4cm, xshift=0 cm]r.center) -- ([yshift = 0cm, xshift=-0.4 cm]n0.center) -- ([yshift = 0.45cm, xshift=0.35 cm] n01.center) -- ([yshift = -0.3cm, xshift=-0.2 cm]n100.center) -- ([yshift = -0.3cm, xshift=0.2 cm]n101.center) -- ++(55:1.3cm) -- ([yshift = -0cm, xshift=0.4 cm] n11.center) -- ([yshift = 0.3cm, xshift=0.3 cm] n1.center) -- cycle; 

\filldraw[fill = yellow!50, opacity=0.4] ([yshift=0.4cm]n11.center) -- ([yshift = -0.3cm, xshift=-0.5 cm]n110.center) -- ([yshift = -0.3cm, xshift=0.5 cm]n111.center) -- cycle; 

\filldraw[fill = brown!50, opacity=0.4] ([yshift=0.4cm]n01.center) -- ([yshift = -0.3cm, xshift=-0.5 cm]n010.center) -- ([yshift = -0.3cm, xshift=0.5 cm]n011.center) -- cycle; 

\filldraw[fill = green!30, opacity=0.4] ([yshift=0.4cm]n0.center) -- ([yshift = 0.3cm, xshift=-0.3 cm]n00.center) -- ([yshift = -0.3cm, xshift=-0.5 cm]n000.center) -- ([yshift = -0.3cm, xshift=0.2 cm]n001.center) -- ++(55:1.3cm) -- ([yshift = 0cm, xshift=0.4 cm]n01.center) -- cycle; 
\end{pgfonlayer}
      \end{tikzpicture}
    \caption{Tree of a tree decomposition. Tree nodes representing special bags are coloured blue, the corresponding maximal clean triangles are indicated as coloured areas.
    The union of all special bags are a center of the graph, with the graphs induced by the maximal clean triangles as petals.}\label{figure:decomposition-and-triangles}
  \end{figure}

Now, an indicator structure $\calB\in\calS(G, C, w, d')$  can be first-order defined as follows. 
Clearly, maximal clean triangles can be easily first-order defined from the relation $S$. For each maximal clean triangle $\delta = (i_0, i_1, i_2)$ with at least two inner nodes, the relation $J$ contains the tuple $\tpl v(\delta)$, and the relation $R_\tau$ contains $v(\delta)$ if and only if is $\tpl v(\delta) \in D_\tau$.
We translate the \MSO formula $\psi$ to a $C$-restricted \MSO formula $\chi'$ such that $\calB \models \psi \Leftrightarrow (G,\aux) \models \chi'$, where $\aux$ is the auxiliary database stored by the dynamic program.
This translation is basically as described by Lemma~\ref{lemma:fointerpreting}. 
The formula $\chi'$ results from $\psi$ by 
\begin{itemize}
 \item $C$-restricting every quantified subformula, so, for example, replacing every quantified subformula $\exists X \; \theta$ by $\exists X\; (\forall x (X(x) \loif C(x)) \land \theta)$ and every quantified subformula $\forall X \; \theta$ by $\forall X\; (\forall x (X(x) \loif C(x)) \loif \theta)$, and
 \item replacing every atom $A(\tpl x)$ by $\theta_A(\tpl x)$, where $\theta_A$ is the first-order formula that defines $A$ in $(G, \aux)$. 
\end{itemize}
It clearly holds that $(G,\aux) \models \chi' \Leftrightarrow \calB \models \psi$, and by Theorem \ref{theorem:center} also $(G,\aux) \models \chi' \Leftrightarrow G \models \varphi$.

We now define a relation $\text{Sub}$ that encodes subsets of $C$.
We observe that $C$ is of size at most $b \log n$ for some $b \in \N$. Thus a subset $C'$ of $C$ can be represented by a tuple $(a_1, \ldots, a_b)$ of nodes, where an element $c \in C$ is in $C'$ if and only if $c$ is the $m$-th element of $C$ with respect to the mapping defined by $N$, $m = (\ell -1) \log n + j$ and the $j$-th bit of $a_\ell$ is one.
By Proposition \ref{prop:log} we finally obtain a first-order formula $\chi$ such that $(G,\aux,\text{Sub}) \models \chi \Leftrightarrow (G,\aux) \models \chi' \Leftrightarrow G \models \varphi$.
That means that a dynamic program that maintains the auxiliary relations as intended can maintain the query $\problem{MC}_\varphi$.

It thus remains to describe how the auxiliary relations can be initialised and updated. 
The set $C$ is initially the bag $B(r)$ of the root of $T$ plus one inner node $v(r,r,r)$, and $S$ contains the tuple $\tpl v(r)$. 

The relations $D_\tau$ are computed in $d_{\text{tree}} \log n$ inductive steps, each of which can be defined in first-order logic, and therefore this computation can be carried out in $\AC^1$, thanks to $\Ind{\log n} = \AC^1$.  More precisely, the computation of the relations $D_\tau$ proceeds inductively in a bottom-up fashion. It starts with triangles  $\delta = (i_0, i_1, i_2)$ for which $T(\delta)$ has exactly one or two inner tree nodes (i,.e., nodes different from  $i_0, i_1, i_2$). Since such graphs $G(\delta)$ have at most $5\ell$ nodes, their type can be determined by a first-order formula.\footnote{Basically, all isomorphism types of such graphs and their respective \MSO types can be directly encoded into first-order formulas.} 

For larger triangles, several cases need to be distinguished. Here we explain the case of a triangle $\delta = (i_0, i_1, i_2)$, for which $i_0$ has child nodes $i'_1$ and $i'_2$ such that $i'_1\preceq i_1$ and $i'_2\preceq i_2$ (cf., Figure~\ref{figure:triangle-ind}). 
In this case, the type $\tau$ of $(G(\delta),\tpl v(\delta))$ can be determined from the types $\tau_1$ of $(G(\delta_1),\tpl v(\delta_1))$ and $\tau_2$ of $(G(\delta_2),\tpl v(\delta_2))$, where $\delta_1=(i'_1, i_1, i_1)$ and $\delta_2=(i'_2, i_2, i_2)$, and the type $\tau_0$ of the graph $G_0$ that includes all edges of $G(\delta)$ that are not already in $G(\delta_1)$ or $G(\delta_2)$. 
More precisely, $\tau_0$ is the type of $(G_0, \tpl v(i_0),\tpl v(i_1),\tpl v(i_2),\tpl v(i'_1),\tpl v(i'_2))$ and $G_0$ is the subgraph of $G$ with node set $V_0 = \bigcup \{B(i_0), B(i_1), B(i_2), B(i'_1), B(i'_2) \}$ and all edges from $G[V_0]$ that have at least one endpoint in $B(i'_1) \cup B(i'_2)$. These types are either already computed or the graphs are of size at most $5\ell$ and their type can therefore be determined by a first-order formula as before. 

We make this more precise. We observe that $(G(\delta),\tpl v(\delta))$ can be composed from the graphs $(G_0, \tpl v(i_0),\tpl v(i_1),\tpl v(i_2),\tpl v(i'_1),\tpl v(i'_2))$, $(G(\delta_1),\tpl v(\delta_1))$ and $(G(\delta_2),\tpl v(\delta_2))$ by first taking the disjoint union of these graphs and afterwards fusing nodes according to the identities induced by the additional constants.
For both operations, the depth-$d$ \MSO type of the resulting structure only depends on the depth-$d$ \MSO type(s) of the original structure(s) \cite[Theorem 1.5 (ii) and Proposition 3.6]{Makowsky04}.
The type $\tau$ of $(G(\delta),\tpl v(\delta))$ is therefore determined by a finite function $f$ as $\tau=f(\tau_0,\tau_1,\tau_2)$, which can be directly encoded into first-order formulas.

Finally, we describe how a dynamic program can maintain $S, C$, and $N$ for $\log n$ many changes. The relation $D_\tau$ is not adapted during the changes. Whenever an edge $(u,v)$ is inserted to or deleted from $G$, the nodes $u$ and $v$ are viewed as \emph{affected}. For every \emph{affected} node $u$ that is not yet in a bag stored in $S$, a \emph{special} tree node is selected  (in some canonical way, e.g.~always the smallest node with respect to $\leq$ is selected) such that $u \in B(j)$. Furthermore, if node $i$ is the least common ancestor of $j$ and another special node it becomes special, as well. It is easy to see that when selecting $j$ as a special node, at most one further 
node becomes special. The tuples $\tpl v(j)$ and $\tpl v(i)$ (if $i$ exists) are added to $S$, their elements are added to $C$, the identifier nodes in $C$ for maximal clean triangles are corrected, and $N$ is updated accordingly.
\end{proofof}

 \section{MSO Optimisation Problems} \label{section:mso-opt}
\newcommand{\enc}[1]{\ensuremath{[#1]}}
\newcommand{\OPT}{\problem{OPT}}
\newcommand{\Sub}{\ensuremath{\text{Sub}}}

With the techniques presented in the previous section, also $\MSO$ definable optimisation problems can be maintained in \DynFO for graphs with bounded treewidth. An \MSO definable optimisation problem $\OPT_\varphi$ is induced by an MSO formula $\varphi(X_1, \ldots, X_m)$ with free set variables $X_1, \ldots, X_m$. Given a graph $G$ with vertex set $V$, it asks for sets $A_1, \ldots, A_m \subseteq V$ of minimal\footnote{The adaptation to maximisation problems is straightforward.} size $\sum_{i=1}^m |A_i|$ such that $G\models \varphi(A_1, \ldots, A_m)$.  Examples (with $m=1$) for such problems are the vertex cover problem and the dominating set problem.

We require from a dynamic program for such a problem that it maintains unary query relations $Q_1,\ldots,Q_m$  that store, at any time, an optimal solution for the current graph.

\begin{theorem} \label{theorem:mso-opt}
For every \MSO formula $\varphi(X_1, \ldots, X_m)$ and every $k$, $\problem{OPT}_{\varphi}$ is in \DynFO for graphs with treewidth at most $k$.
\end{theorem}

As already mentioned in the previous section, for every $k$ and every $\GSO$ formula $\varphi$ there is an \MSO formula $\psi$ that is equivalent on graphs with treewidth $k$ \cite[Theorem 2.2]{Courcelle94}. Moreover, if $\varphi = \exists X_1 \cdots \exists X_m \, \varphi'$, then $\psi$ is of the form $\exists X_1^1 \cdots \exists X_1^\ell \cdots \exists X_m^1 \cdots \exists X_m^\ell \psi'$, for some natural number $\ell$.
So, we can conclude the following corollary.

\begin{corollary} \label{corr:gso-opt}
For every \GSO formula $\varphi(X_1, \ldots, X_m)$ and every $k$, $\problem{OPT}_{\varphi}$ is in \DynFO for graphs with treewidth at most $k$.
\end{corollary}
Given the machinery from the previous section, the plan for a dynamic program for an MSO-definable optimisation problem is relatively straightforward. Again, it suffices to show $(\AC^1,\log n)$-maintainability. The affected nodes of the graph after $\log n$ changes are again collected in a center $C$ of the graph (with $\bigO(\log n)$ additional nodes as before). For each petal $D_i$ and each relevant MSO-type $\tau$ we basically maintain a collection $(B_1,\ldots,B_m)$ of subsets of $D_i-C$ that yields type $\tau$ in $D_i$ and is  minimal with respect to~$\sum_{i=1}^m |B_i|$. Then, it is easy to compute in a first-order fashion, for every possible colouring of~$C$, the minimum achievable overall sum for extensions of the colouring that make $\varphi$ true.

\begin{proofof}{Theorem \ref{theorem:mso-opt}}
We only prove the special case of $m=1$, the extension to the general case is straightforward.
Let $\varphi(X)$ be an \MSO formula of quantifier depth $d$. 
The proof of Theorem~\ref{theorem:mso} shows how one can obtain a dynamic program that $(\AC^1, \log n)$-maintains the model checking problem $\problem{MC}_\psi$ for $\psi \df \exists X\, \varphi$. We adapt this proof, and reuse its notation, in order to obtain such a dynamic program for $\problem{OPT}_{\varphi}$, using almost the same auxiliary relations. 
Together with Theorem \ref{theorem:fewerChanges} and Proposition \ref{prop:msoadi}, the result follows. 

In the following, we sketch the proof idea. 
We consider $\varphi$ to be an \MSO sentence over the signature $\{E,X\}$.
Let $G^{+X} = (V,E,X)$ be an arbitrary expansion of a graph $G$ with a center $C$ of connection width $w$, for some constant $w$. By Theorem \ref{theorem:center} there is a number $d'$ and an \MSO sentence $\psi$ such that for every $\calB^{+X} \in \calS(G^{+X},C,w,d')$ it holds that $G^{+X} \models \varphi$ if and only if $\calB^{+X} \models \psi$.
So, the formula $\psi$ uses the type information on the petals provided by $\calB^{+X}$ as well as $\restrict{G^{+X}}{C}$ directly to check whether the relation $X$ represents a \emph{feasible} solution of the problem $\OPT_\varphi$.
In the proof of Theorem \ref{theorem:mso} we explained how to obtain a first-order formula $\chi$ from $\psi$ such that $(G, \aux, \text{Sub}) \models \chi \Leftrightarrow \calB \models \psi$, for the auxiliary database $\aux$ maintained by the dynamic program constructed in the proof of Theorem \ref{theorem:mso} and a relation $\text{Sub}$ encoding subsets of $C$. 

Let $\calB \in \calS(G,C,w,d'+1)$ be an \MSO indicator structure for $G$. Our goal is to maintain some relations that augment the type information provided by $\calB$ such that a formula $\chi'$ similar to $\chi$ can ``guess'' a relation $X$, check that it is a feasible solution, compute its size and verify that no feasible solution of smaller size exists.
Of course, a relation $X$ of unrestricted size cannot be quantified in first-order logic, even in the presence of $\Sub$, but we will see that the restriction of $X$ to $C$ and the type information on the petals can be quantified, which is sufficient for our purpose.

We now give the details of the construction.
The structure $\calB$ contains relations $R_\tau$ such that $u^i_1 \in R_\tau$ if and only the depth-$(d'+1)$ \MSO type of $(\calA_i, \tpl u^i)$ over signature $\Sigma = \{E, c_1, \ldots, c_{3\ell+1}\}$ is $\tau$, where the subgraph $\calA_i$ over universe $D_i$ and the tuple $\tpl u^i$ are defined as in Subsection~\ref{subsction:fv}.
We say that a depth-$d'$ type $\tau'$ over signature $\Sigma^{+X} \df \Sigma \cup \{X\}$ \emph{can be realised} in $(\calA_i, \tpl u^i)$ by a set $A_i \subseteq D_i$, if $\tau'$ is the depth-$d'$ \MSO type of $(\calA_i, \tpl u^i, A_i)$. If $(\calA_i, \tpl u^i)$ has depth-$(d'+1)$ \MSO type $\tau$, the existence of such a set is equivalent to the statement $\exists X\, \alpha_{\tau'} \in \tau$.
We note that $\tau'$ already determines whether $u_j^i \in A_i$ shall hold, for each constant $u_j^i$ from the tupel $\tpl u^i$.

The dynamic program maintains relations $\#R_{\tau'}$ and $Q_{\tau'}$, for each depth-$d'$ \MSO type over $\Sigma^{+X}$. 
The relations $\#R_{\tau'}$ give the minimal size of a set that realises the type $\tau'$.
So, if $\tau'$ can be realised in $(\calA_i,\tpl u^i)$ by some set $A$, and $s$ is the minimal size of such a set, then $\#R_{\tau'}$ shall contain the tuple $(u^i_1, v_s)$, where $v_s$ is the $(s+1)$-th element\footnote{We ignore the case that the size could be as large as $|V|$, which can be handled by some additional encoding.} with respect to~$\leq$.
Furthermore, for the lexicographically minimal set $A$ of this kind and size $s$, $Q_{\tau'}$ shall contain all tuples $(u^i_1,a)$, where $a \in A$.

We construct a first-order formula $\chi'$ from $\chi$ that is able to define an \emph{optimal} solution $X$ for $\OPT_\varphi$ from $(G, \aux, \Sub)$ expanded by the relations $\#R_{\tau'}$ and $Q_{\tau'}$.
First, this formula quantifies for each depth-$d'$ \MSO type $\tau'$ the set of identifiers $u^i_1$ such that $X$ realises $\tau'$ in $(\calA_i, \tpl u^i)$ and checks consistency: as for each node $v \in C$ that appears in $\tpl u^i$ the type $\tau'$ already determines whether $v \in X$ shall hold, the respective types need to agree for nodes that appear in multiple petals.
For each $u^i_1$ the assigned type also needs to be realisable in the respective substructure $(\calA_i, \tpl u^i)$, which can be checked using the relations $R_\tau$ of $\calB$.
Using this information, $\chi'$ can apply $\chi$ to check that the implied set $X$ is a feasible solution.
With the help of $\#R_\tau$ it can compute the size of $X$, as \FOar is able to add up logarithmically many numbers \cite[Theorem 1.21]{Vollmer99} and $C$ is only of logarithmic size in $|V|$.
Also $\chi'$ checks that no other assignment of types $\tau'$ to identifier nodes results in feasible solutions of smaller size.
Finally, $\chi'$ uses the relations $Q_{\tau'}$ to actually return an optimal solution $X$.

Building on the proof of Theorem \ref{theorem:mso}, it remains to show that the additional auxiliary relations $\#R_{\tau'}$ and $Q_{\tau'}$ can be initialised and maintained. Actually, we maintain similarly defined relations $\#D_{\tau'}$ and $F_{\tau'}$, the relations $\#R_{\tau'}$ and $Q_{\tau'}$ are then first-order definable by the dynamic program using these relations.

Let $\delta$ be a triangle such that $G(\delta)$ has at least one inner node. 
Similar to the relations $D_\tau$ used in the proof of Theorem \ref{theorem:mso}, here a relation $\#D_{\tau'}$ contains the tuple $(\tpl v(\delta), u)$ if and only if (1) the depth-$d'$ \MSO type $\tau'$ is realisable in $(G(\delta),\tpl v(\delta))$, and (2)  $u$ is the $(s+1)$-th element with respect to $\leq$, where $s$ is the minimal size of a set that realises this type.
Furthermore, for the lexicographically minimal set $A$ of this kind and size $s$, $F_{\tau'}$ contains all tuples $(\tpl v(\delta),a)$, where $a \in A$.
It is clear that these relations suffice to define the relations $\#R_{\tau'}$ and $Q_{\tau'}$, given the other relations of the proof of Theorem \ref{theorem:mso}.

The proof of Theorem \ref{theorem:mso} can be extended to show that the initial versions of these auxiliary relations can be computed in $\AC^1$. For the inductive step of this computation, a type $\tau'$ realisable in a structure $(G(\delta), \tpl v(\delta))$ might be achievable by a finite number of combinations of types of its substructures. Here, the overall size of the realising set for $X$ needs to be computed and the minimal solution needs to be picked. This is possible by a $\FOar$-formula since the number of possible combinations is bounded by a constant depending only on $d$ and $k$.

The updates of the auxiliary relations are exactly as in the proof of Theorem~\ref{theorem:mso}. Since $D_\tau$ needs no updates there, neither $\#D_{\tau'}$  nor $F_{\tau'}$  do, here.
\end{proofof}

From the proof it is easy to see that a dynamic program can also maintain the \emph{size} $s$ of an optimal solution, either implicitly as $\sum_{j=1}^m|Q_j|$ for distinguished relations $Q_j$, or as $\{v_s\}$. 
Additionally, it can easily be adapted for optimisation problems on \emph{weighted} graphs, where nodes and edges have polynomial weights in $n$.

\section{Conclusion}\label{section:conclusion}
In this paper, we introduced a strategy for maintaining queries by periodically restarting its computation from scratch and limiting the number of change steps that have to be taken into account. This has been captured in the notion of $(\calC, f)$-maintainable queries, and we proved in particular that all $(\AC^1, \log n)$-maintainable, almost domain independent queries are actually in \DynFO. As a consequence, decision and optimisation queries definable in  \MSO- and \GSO-logic are in $\DynFO$ for graphs of bounded treewidth.  For this, we stated a Feferman-Vaught-type composition theorem for these logics, which might be interesting in its own right. Though we phrase our results for \MSO and \GSO for graphs only, their proofs translate swiftly to general relational structures. 

Apart from this paper, this strategy is already used in \cite{DattaMVZ18} to prove that the reachability query can be maintained dynamically under insertions and deletions of a non-constant number of edges per change step.

We believe that our strategy will find further applications. For instance, it is conceivable that interesting queries on planar graphs, such as the shortest-path query, can be maintained for a bounded number of changes using auxiliary data computed by an $\AC^1$ algorithm (in particular since many important data structures for planar graphs can be constructed in logarithmic space and therefore in $\AC^1$).

\section*{Acknowledgment}
  \noindent We thank an anonymous referee for valuable suggestions that greatly simplified the proof of Theorem \ref{theorem:center}.
   
 \bibliographystyle{alpha}  
 \bibliography{bibliography}

\end{document}